\documentclass[acmsmall,nonacm]{acmart}
\acmJournal{PACMPL}
\acmVolume{1}
\acmNumber{OOPSLA}
\acmArticle{1}
\acmYear{2020}
\acmMonth{1}
\acmDOI{}
\startPage{1}

\setcopyright{none}

\usepackage{booktabs}
\usepackage{subcaption}
\usepackage{amsmath}
\usepackage{mathpartir}
\usepackage[all,cmtip]{xy}
\usepackage[nameinlink,noabbrev,capitalise]{cleveref}
\usepackage{csquotes}
\usepackage{booktabs}
\usepackage{microtype}
\usepackage{stmaryrd}
\usepackage{graphicx}
\usepackage{enumitem}
\usepackage{mathtools}
\usepackage{courier}
\usepackage{bussproofs}
\usepackage{tikz-cd}
\usepackage{diagbox}
\usepackage{cmll}
\usepackage{afterpage}
\usepackage[normalem]{ulem}
\usepackage{cancel}
\usepackage{dashbox}
\usepackage{mdframed}
\usepackage[mathscr]{euscript}
\DeclareMathAlphabet\mathbfscr{U}{eus}{b}{n}
\usepackage{colortbl}
\usepackage{anyfontsize}
\usepackage{listings}
\usepackage{alltt} 
\usepackage{multicol} 
\usepackage{scalerel} 
\usepackage{tipa}
\usepackage{adjustbox}
\usepackage{thm-restate}
\theoremstyle{remark}
\newtheorem{remark}{Remark}
\theoremstyle{acmplain}

\usepackage[]{todonotes}

\newcommand{\dominic}[2][]{\todo[inline,color=red!40,author=Dominic,#1]{#2}}

\newcommand{\dnote}[1]{\dominic{#1}{}}

\newcommand{\costmonad}[2]{\mathbb{M}\,#1\,#2}
\newcommand{\potentialtype}[2]{[#1] #2}
\newcommand{\tyunitp}{\mathbf{1}}
\newcommand{\tyunitn}{\top}
\newcommand{\tyvoid}{\mathbf{0}}
\newcommand{\tyvec}[2]{\mathtt{Vec}\ #1\ #2}
\newcommand{\tyqueue}[2]{\mathtt{Queue}\ #1\ #2}

\newcommand{\sub}{<:}

\newcommand{\tmunit}{\mathtt{unit}}
\newcommand{\tminl}[1]{\mathtt{inl}(#1)}
\newcommand{\tminr}[1]{\mathtt{inr}(#1)}
\newcommand{\tmtensor}[2]{\langle #1, #2 \rangle}
\newcommand{\tmtick}[1]{\ \uparrow^{#1}\ }
\newcommand{\tmret}[1]{\mathtt{ret}\ #1}
\newcommand{\tmstore}[1]{\mathtt{store}\ #1}
\newcommand{\tmrelease}[3]{\mathtt{release}\ #1 = #2\ \mathtt{in}\ #3}
\newcommand{\tmbind}[3]{\mathtt{bind}\ #1 = #2\ \mathtt{in}\ #3}
\newcommand{\tmcase}[2]{#1.\mathtt{case}\ \lbrace #2 \rbrace}
\newcommand{\tmcocase}[1]{\mathtt{cocase}\ \lbrace #1 \rbrace}
\newcommand{\tmfst}{\mathtt{fst}}
\newcommand{\tmsnd}{\mathtt{snd}}
\newcommand{\tmrun}[1]{\mathtt{run}(#1)}

\newcommand{\tmpay}[1]{\mathtt{pay}(#1)}
\newcommand{\tmsplit}[1]{\mathtt{split}(#1)}
\newcommand{\tmplet}[3]{\mathtt{plet}\ #1 = #2\ \mathtt{in}\ #3}
\newcommand{\tmclet}[3]{\mathtt{clet}\ #1 = #2\ \mathtt{in}\ #3}

\newcommand{\evaluatesto}{\Downarrow}
\newcommand{\forcesto}[1]{\Downarrow^{#1}}

\newcommand{\potential}{p}
\newcommand{\cost}{\kappa}
\newcommand{\kripkemodel}[1]{\llbracket #1 \rrbracket}
\newcommand{\kripkemodele}[1]{\llbracket #1 \rrbracket_{\mathcal{E}}}

\newcommand{\op}{\mathsf{op}}
\newcommand{\setcategory}{\mathbf{Set}}
\newcommand{\costcategory}{\mathbb{C}}
\newcommand{\modelcategory}{\mathbb{X}}
\newcommand{\copresheafmodel}{\hat{\mathbb{X}}}
\newcommand{\kripkecategory}{\mathbb{K}}

\newcommand{\homset}[3]{\mathbf{Hom}_{#1}(#2,#3)}
\newcommand{\interpret}[1]{\llbracket #1 \rrbracket}
\newcommand{\yoneda}[1]{\mathcal{Y}(#1)}
\newcommand{\catobjects}[1]{\mathbf{Obj}(#1)}
\newcommand{\idmorph}[1]{\mathrm{id}_{#1}}

\newcommand{\dayconvolution}{\otimes_{\mathrm{Day}}}
\newcommand{\dayidentity}{I_{\mathrm{Day}}}
\newcommand{\dayexponential}{\multimap_\mathrm{Day}}

\newcommand{\lean}{\includegraphics[scale=0.01, trim = 0px 125px 0px 0px]{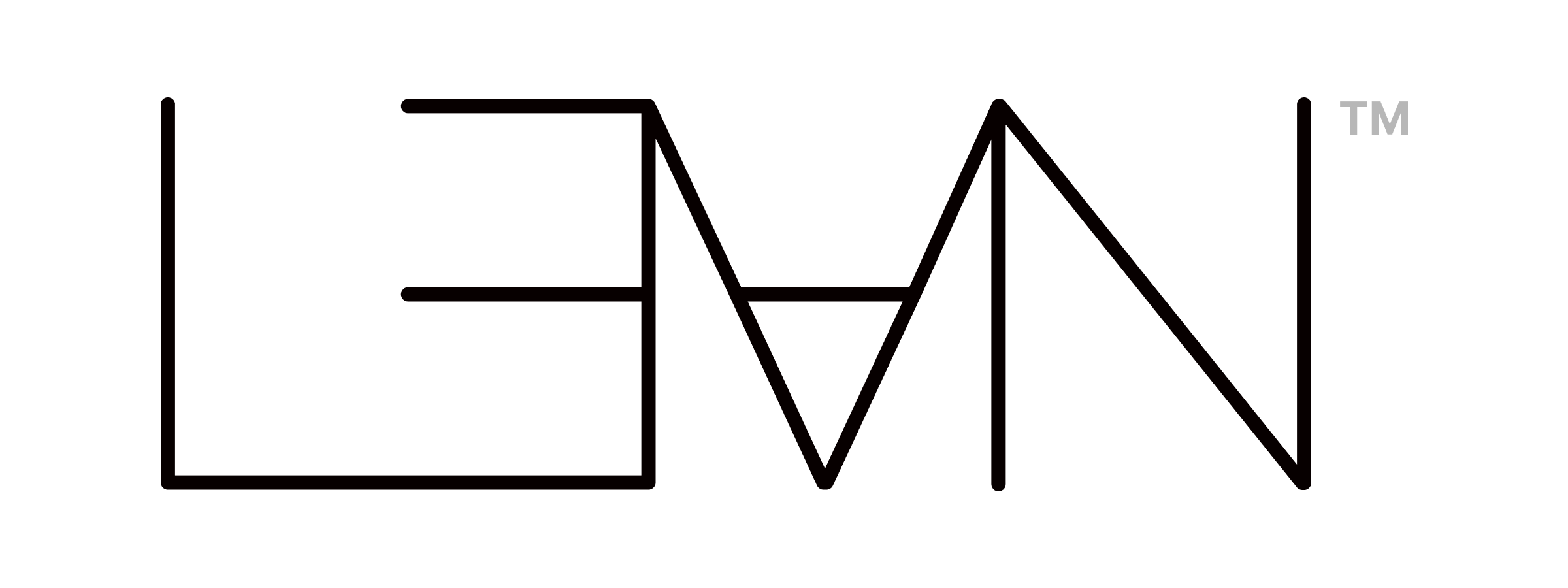}}

\newcommand{\tyrulevar}{
    \begin{prooftree}
      \AxiomC{\phantom{$\Gamma$}}
      \RightLabel{\textsc{Var}}
      \UnaryInfC{$x : \tau \vdash x : \tau$}
    \end{prooftree}
}

\newcommand{\tyrulesub}{
    \begin{prooftree}
      \AxiomC{$\Gamma_2 \vdash e : \sigma$}
      \AxiomC{$\Gamma_1 \sub \Gamma_2$}
      \AxiomC{$\sigma \sub \tau$}
      \RightLabel{\textsc{Sub}}
      \TrinaryInfC{$\Gamma_1 \vdash e : \tau$}
    \end{prooftree}
}

\newcommand{\tyrulewk}{
    \begin{prooftree}
      \AxiomC{$\Gamma \vdash e : \tau$}
      \RightLabel{\textsc{Wk}}
      \UnaryInfC{$\Gamma, x : \tau' \vdash e : \tau$}
    \end{prooftree}
}

\newcommand{\tyruleex}{
    \begin{prooftree}
      \AxiomC{$\Gamma_1, x : \tau_1, y : \tau_2, \Gamma_2 \vdash e : \tau_3$}
      \RightLabel{\textsc{Ex}}
      \UnaryInfC{$\Gamma_1, y : \tau_2, x : \tau_1, \Gamma_2 \vdash e : \tau_3$}
    \end{prooftree}
}

\newcommand{\tyruleunitp}{
    \begin{prooftree}
      \AxiomC{\phantom{$\Gamma \vdash$}}
      \RightLabel{\textsc{Unit-P}}
      \UnaryInfC{$\vdash \tmunit : \tyunitp$}
    \end{prooftree}
}

\newcommand{\tyrulecaseunitp}{
    \begin{prooftree}
      \AxiomC{$\Gamma_1 \vdash e_1 : \tyunitp$}
      \AxiomC{$\Gamma_2 \vdash e_2 : \tau$}
      \RightLabel{\textsc{Case-Unit-P}}
      \BinaryInfC{$\Gamma_1, \Gamma_2 \vdash \tmcase{e_1}{\tmunit \Rightarrow e_2} : \tau$}
    \end{prooftree}
}

\newcommand{\tyruleunitn}{
    \begin{prooftree}
      \AxiomC{}
      \RightLabel{\textsc{Unit-N}}
      \UnaryInfC{$\Gamma \vdash \tmcocase{} : \tyunitn$}
    \end{prooftree}
}

\newcommand{\tyrulevoid}{
    \begin{prooftree}
      \AxiomC{$\Gamma \vdash e : \tyvoid$}
      \RightLabel{\textsc{Void}}
      \UnaryInfC{$\Gamma \vdash \tmcase{e}{} : \tau$}
    \end{prooftree}
}

\newcommand{\tyruleabs}{
  \begin{prooftree}
      \AxiomC{$\Gamma, x : \tau \vdash e : \sigma$}
      \RightLabel{\textsc{Abs}}
      \UnaryInfC{$\Gamma \vdash \lambda x.e : \tau \multimap \sigma$}
    \end{prooftree}
}

\newcommand{\tyruleapp}{
    \begin{prooftree}
      \AxiomC{$\Gamma_1 \vdash e_1 : \tau_1 \multimap \tau_2$}
      \AxiomC{$\Gamma_2 \vdash e_2 : \tau_1$}
      \RightLabel{\textsc{App}}
      \BinaryInfC{$\Gamma_1, \Gamma_2 \vdash e_1\, e_2 : \tau_2$}
    \end{prooftree}
}

\newcommand{\tyruleinl}{
    \begin{prooftree}
      \AxiomC{$\Gamma \vdash e : \sigma$}
      \RightLabel{\textsc{Inl}}
      \UnaryInfC{$\Gamma \vdash \tminl{e} : \sigma \oplus \tau$}
    \end{prooftree}
}

\newcommand{\tyruleinr}{
\begin{prooftree}
      \AxiomC{$\Gamma \vdash e : \tau$}
      \RightLabel{\textsc{Inr}}
      \UnaryInfC{$\Gamma \vdash \tminr{e} : \sigma \oplus \tau$}
    \end{prooftree}
}

\newcommand{\tyrulecasesum}{
    \begin{prooftree}
      \AxiomC{$\Gamma_1 \vdash e : \tau_1 \oplus \tau_2$}
      \AxiomC{$\Gamma_2, x : \tau_1 \vdash e_1 : \tau_3$}
      \AxiomC{$\Gamma_2, x : \tau_2 \vdash e_2: \tau_3$}
      \RightLabel{\textsc{Case-Sum}}
      \TrinaryInfC{$\Gamma_1, \Gamma_2 \vdash \tmcase{e}{\tminl{x} \Rightarrow e_1, \tminr{x} \Rightarrow e_2}: \tau_3$}
    \end{prooftree}
}

\newcommand{\tyruletensor}{
  \begin{prooftree}
    \AxiomC{$\Gamma_1 \vdash e_1 : \tau_1$}
    \AxiomC{$\Gamma_2 \vdash e_2 : \tau_2$}
    \RightLabel{\textsc{Tensor}}
    \BinaryInfC{$\Gamma_1, \Gamma_2 \vdash \tmtensor{e_1}{e_2} : \tau_1 \otimes \tau_2$}
  \end{prooftree}
}

\newcommand{\tyrulecasetensor}{
  \begin{prooftree}
    \AxiomC{$\Gamma_1 \vdash e_1 : \tau_1 \otimes \tau_2$}
    \AxiomC{$\Gamma_2, x : \tau_1, y : \tau_2 \vdash e_2 : \tau_3$}
    \RightLabel{\textsc{Case-Tensor}}
    \BinaryInfC{$\Gamma_2, \Gamma_1 \vdash \tmcase{e_1}{\tmtensor{x}{y} \Rightarrow e_2} : \tau_3$}
  \end{prooftree}
}

\newcommand{\tyrulewith}{
  \begin{prooftree}
    \AxiomC{$\Gamma \vdash e_1 : \tau_1$}
    \AxiomC{$\Gamma \vdash e_2 : \tau_2$}
    \RightLabel{\textsc{With}}
    \BinaryInfC{$\Gamma \vdash \tmcocase{\tmfst \Rightarrow e_1, \tmsnd \Rightarrow e_2} : \tau_1 \with \tau_2$}
  \end{prooftree}
}

\newcommand{\tyrulefst}{
  \begin{prooftree}
      \AxiomC{$\Gamma \vdash e : \tau_1 \with \tau_2$}
      \RightLabel{\textsc{Fst}}
      \UnaryInfC{$\Gamma \vdash e.\tmfst : \tau_1$}
  \end{prooftree}
}

\newcommand{\tyrulesnd}{
  \begin{prooftree}
      \AxiomC{$\Gamma \vdash e : \tau_1 \with \tau_2$}
      \RightLabel{\textsc{Snd}}
      \UnaryInfC{$\Gamma \vdash e.\tmfst : \tau_1$}
  \end{prooftree}
}

\newcommand{\tyrulebind}{
  \begin{prooftree}
    \AxiomC{$\Gamma_1 \vdash e_1 : \costmonad{p_1}{\tau_1}$}
    \AxiomC{$\Gamma_2, x : \tau_1 \vdash e_2 : \costmonad{p_2}{\tau_2}$}
    \RightLabel{\textsc{Bind}}
    \BinaryInfC{$\Gamma_1, \Gamma_2 \vdash \tmbind{x}{e_1}{e_2} : \costmonad{(p_1 + p_2)}{\tau_2}$}
  \end{prooftree}
}

\newcommand{\tyruleret}{
  \begin{prooftree}
      \AxiomC{$\Gamma \vdash e : \tau$}
      \RightLabel{\textsc{Ret}}
      \UnaryInfC{$\Gamma \vdash \tmret{e} : \costmonad{0}{\tau}$}
  \end{prooftree}
}

\newcommand{\tyrulerun}{
  \begin{prooftree}
    \AxiomC{$\Gamma \vdash e : \costmonad{0}{\tau}$}
    \RightLabel{\textsc{Run}}
    \UnaryInfC{$\Gamma \vdash \tmrun{e} : \tau$}
  \end{prooftree}
}

\newcommand{\tyruletick}{
  \begin{prooftree}
    \AxiomC{\rule{0em}{1em}} 
    \RightLabel{\textsc{Tick}}
    \UnaryInfC{$\vdash \tmtick{p} : \costmonad{p}{\tyunitp}$}
  \end{prooftree}
}

\newcommand{\tyrulestore}{
  \begin{prooftree}
    \AxiomC{$\Gamma \vdash e : \tau$}
    \RightLabel{\textsc{Store}}
    \UnaryInfC{$\Gamma \vdash \tmstore{e} : \costmonad{p}{\potentialtype{p}{\tau}}$}
  \end{prooftree}
}

\newcommand{\tyrulerelease}{
  \begin{prooftree}
    \AxiomC{$\Gamma_1 \vdash e_1 : \potentialtype{p_1}{\tau_1}$}
    \AxiomC{$\Gamma_2, x : \tau_1 \vdash e_2 : \costmonad{(p_1 + p_2)}{\tau_2}$}
    \RightLabel{\textsc{Release}}
    \BinaryInfC{$\Gamma_1, \Gamma_2 \vdash \tmrelease{x}{e_1}{e_2} : \costmonad{p_2}{\tau_2}$}
  \end{prooftree}
}

\newcommand{\tyrulepay}{
  \begin{prooftree}
    \AxiomC{$\Gamma \vdash e : \potentialtype{p}{\costmonad{p}{\tau}}$}
    \RightLabel{\textsc{Pay}}
    \UnaryInfC{$\Gamma \vdash \tmpay{e} : \tau$}
  \end{prooftree}
}

\newcommand{\tyrulesplit}{
  \begin{prooftree}
    \AxiomC{$\Gamma \vdash e : \costmonad{p_1 + p_2}{\tau}$}
    \RightLabel{\textsc{Split}}
    \UnaryInfC{$\Gamma \vdash \tmsplit{e} : \costmonad{p_1}{(\costmonad{p_2}{\tau})}$}
  \end{prooftree}
}

\newcommand{\tyruleplet}{
  \begin{prooftree}
    \AxiomC{$\Gamma_1 \vdash e_1 : \potentialtype{p}{\tau_1}$}
    \AxiomC{$\Gamma_2, x : \tau_1 \vdash e_2 : \tau_2$}
    \RightLabel{\textsc{PLet}}
    \BinaryInfC{$\Gamma_1,\Gamma_2 \vdash \tmplet{x}{e_1}{e_2} : \potentialtype{p}{\tau_2}$}
  \end{prooftree}
}

\newcommand{\eqrulerefl}{
  \begin{prooftree}
    \AxiomC{$\mathcal{D} : (\Gamma \vdash e : \tau)$}
    \RightLabel{$\equiv$\textsc{-Refl}}
    \UnaryInfC{$\mathcal{D} \equiv \mathcal{D}$}
  \end{prooftree}
}

\newcommand{\eqrulesymm}{
  \begin{prooftree}
    \AxiomC{$\mathcal{D}_1 \equiv \mathcal{D}_2$}
    \RightLabel{$\equiv$\textsc{-Sym}}
    \UnaryInfC{$\mathcal{D}_2 \equiv \mathcal{D}_1$}
  \end{prooftree}
}

\newcommand{\eqruletrans}{
  \begin{prooftree}
    \AxiomC{$\mathcal{D}_1 \equiv \mathcal{D}_2$}
    \AxiomC{$\mathcal{D}_2 \equiv \mathcal{D}_3$}
    \RightLabel{$\equiv$\textsc{-Trans}}
    \BinaryInfC{$\mathcal{D}_1 \equiv \mathcal{D}_3$}
  \end{prooftree}
}

\newcommand{\eqruleconginl}{
  \begin{prooftree}
    \AxiomC{$\mathcal{D}_1 \equiv \mathcal{D}_2$}
    \RightLabel{$\equiv$\textsc{-Cong-Inl}}
    \UnaryInfC{$\textsc{Inl}(\mathcal{D}_1) \equiv \textsc{Inl}(\mathcal{D}_2)$}
  \end{prooftree}
}

\newcommand{\eqruleconginr}{
  \begin{prooftree}
    \AxiomC{$\mathcal{D}_1 \equiv \mathcal{D}_2$}
    \RightLabel{$\equiv$\textsc{-Cong-Inr}}
    \UnaryInfC{$\textsc{Inr}(\mathcal{D}_1) \equiv \textsc{Inr}(\mathcal{D}_2)$}
  \end{prooftree}
}

\newcommand{\eqrulecongfst}{
  \begin{prooftree}
    \AxiomC{$\mathcal{D}_1 \equiv \mathcal{D}_2$}
    \RightLabel{$\equiv$\textsc{-Cong-Fst}}
    \UnaryInfC{$\textsc{Fst}(\mathcal{D}_1) \equiv \textsc{Fst}(\mathcal{D}_2)$}
  \end{prooftree}
}

\newcommand{\eqrulecongsnd}{
  \begin{prooftree}
    \AxiomC{$\mathcal{D}_1 \equiv \mathcal{D}_2$}
    \RightLabel{$\equiv$\textsc{-Cong-Snd}}
    \UnaryInfC{$\textsc{Snd}(\mathcal{D}_1) \equiv \textsc{Snd}(\mathcal{D}_2)$}
  \end{prooftree}
}

\usepackage{xcolor}

\begin{document}

\title{Categorical Models of Amortized Cost}
\subtitle{An Adjoint Relationship between Cost and Potential}

\keywords{cost analysis, amortized complexity, categorical models}

\ccsdesc[300]{Theory of computation~Lambda calculus}
\ccsdesc[300]{Theory of computation~Type theory}

\author{David Binder}
\orcid{0000-0003-1272-0972}
\affiliation{
  \department{Department of Computer Science}
  \institution{University of Kent}
  \streetaddress{Kennedy Building}
  \city{Canterbury}
  \postcode{CT2 7FS}
  \country{United Kingdom}
}
\email{D.Binder@kent.ac.uk}

\author{David Corfield}
\orcid{0000-0003-0432-3221}
\affiliation{
  \institution{Independent Researcher}
  \city{Canterbury}
  \country{United Kingdom}
}

\email{D.Corfield@kent.ac.uk}

\author{Dominic Orchard}
\orcid{0000-0002-7058-7842}
\affiliation{
  \department{Department of Computer Science}
  \institution{University of Kent}
  \streetaddress{Kennedy Building}
  \city{Canterbury}
  \postcode{CT2 7FS}
  \country{United Kingdom}
}
\affiliation{
  \department{Department of Computer Science and Technology}
  \institution{University of Cambridge}
  \streetaddress{William Gates Building}
  \city{Cambridge}
  \country{United Kingdom}
}
\email{D.A.Orchard@kent.ac.uk}

\author{Vineet Rajani}
\orcid{0000-0001-7701-8311}
\affiliation{
  \department{Department of Computer Science and Engineering}
  \institution{UNSW Sydney}
  \city{Sydney}
  \country{Australia}
}
\email{v.rajani@unsw.edu.au}

\begin{abstract}
  Various type systems have been developed to track the cost $\kappa$ of a computation using a cost-tracking monad $\costmonad{\kappa}{\tau}$.
  On its own, this only tracks the \emph{worst-case} cost of a computation.
  If we also want to track \emph{amortized cost}, then we can add a type $\potentialtype{\kappa}{\tau}$ which stores \emph{potential} $\kappa$ with a type $\tau$, together with operations for storing and releasing potential.
  In this work, we build on one such system, $\lambda$-amor:
  $\lambda$-amor allows to track cost and potential in the type system and subsumes effect and coeffect-based systems, call-by-value and call-by-name based languages.
  In this paper, we identify the abstract properties that denotational models of type theories for cost and potential have to satisfy:
  Cost and potential must be modelled by an adjoint pair of graded functors, where the functor modelling cost forms both a graded monad and a compatible graded comonad.
  We present three concrete instances of this general abstract scheme:
  (1) A simple set-theoretic model that ignores the cost tracked by the type system,
  (2) the Kripke logical relations model in the original $\lambda$-amor paper (which we show can be turned into an instance of the adjoint model),
  and (3) a novel model based on copresheaves on a monoidal category of costs, where we model pairs and functions by Day convolution and its right-adjoint.
\end{abstract}

\maketitle

\section{Introduction}
\label{sec:intro}
Due to their compositional nature, type systems are widely used to ensure different correctness properties and to track various aspects of program behavior.
One program behaviour that we often want to reason about is \emph{computational cost}; specifically, we want to ensure that a program does not need more time, memory or other computational resources than we expect.
In this article we build on the approach of tracking cost
via a \emph{graded monad}, originally proposed by \citet{Danielsson2008semiformal}, which
has subsequently been built on by others~\cite{Gaboardi2016combining,Radivcek2017monadic,Gaboardi2021gradedhoare}. Graded monads generalise the single endofunctor
of a monad to a family of endofunctors indexed by a monoidal structure, whose operations
are then stratified by the monoidal structure $(\mathbb{M}, +, 0)$. The result, when rendered in a type system, is a pair of constructs:

\begin{minipage}{0.3\textwidth}
  \begin{prooftree}
    \AxiomC{$\Gamma \vdash t : \tau$}
    \RightLabel{\textsc{Return}}
    \UnaryInfC{$\Gamma \vdash \tmret{t} : \costmonad{0}{\tau}$}
  \end{prooftree}
\end{minipage}
\begin{minipage}{0.6\textwidth}
  \begin{prooftree}
    \AxiomC{$\Gamma \vdash t_1 : \costmonad{p_1}{\tau_1}$}
    \AxiomC{$\Gamma, x : \tau_1 \vdash t_2 : \costmonad{p_2}{\tau_2}$}
    \RightLabel{\textsc{Bind}}
    \BinaryInfC{$\Gamma \vdash \tmbind{x}{t_1}{t_2} : \costmonad{(p_1 + p_2)}{\tau_2}$}
  \end{prooftree}
\end{minipage}
\vspace{0.1cm}

\noindent
where $\tmret{t}$ injects a pure term of type $\tau$ into the graded monad at `grade' 0,
witnessing the purity of the computation (cost free), and
$\tmbind{x}{t_1}{t_2}$ sequences two computations taking the addition of their costs.
Graded monads have many applications, but we focus here on their use for tracking
computation costs. Such systems further add an operation for incurring a non-zero cost, which here we present as $\tmtick{\kappa}$ (pronounced \enquote{tick}), incurring
a cost $\kappa$ and returning a term of the unit type:
\begin{prooftree}
    \AxiomC{}
    \RightLabel{\textsc{Tick}}
    \UnaryInfC{$\Gamma \vdash \tmtick{\kappa} : \costmonad{\kappa}{\tyunitp}$}
\end{prooftree}

This approach to cost tracking is often used to provide worst-case guarantees: if we run a program
of type $\costmonad{\kappa}{\tau}$ then we expect it to return a value of type $\tau$ and incur at
most cost $\kappa$ through the ticks which we have used to annotate the program. However, only
worst-case analysis is often not good enough, especially when we analyze algorithms on functional
data structures, where we want to track the \emph{amortized cost} of a
computation. One way to establish amortized cost bounds is via the \emph{method of potentials}
\cite{Tarjan1985amortized}: Cheap and frequent operations allow to \emph{store some potential},
while expensive and rare operations can \emph{release stored potential}. \citet{Rajani2021lamor}
reflect this sort of reasoning in the type system by using an \emph{indexed potential modality}
$\potentialtype{\kappa}{\tau}$. This type describes the terms of type $\tau$ which store a potential
of $\kappa$ units with them. Since potential is a resource which is used to account for cost, they
are not allowed to be freely duplicated. \citet{Rajani2021lamor} handle this by making their type
theory affine. This is reflected in the typing of store and release.


\hspace{-0.5em}\begin{minipage}{0.38\linewidth}
\begin{prooftree}
    \AxiomC{$\Gamma \vdash t : \tau$}
    \RightLabel{\textsc{Store}}
    \UnaryInfC{$\Gamma \vdash \tmstore{t} : \costmonad{\kappa}{(\potentialtype{\kappa}{\tau})}$}
\end{prooftree}
\end{minipage}
\begin{minipage}{0.59\linewidth}
\begin{prooftree}
    \AxiomC{$\Gamma_1 \vdash t_1 : \potentialtype{\kappa_1}{\tau_1}$}
    \AxiomC{$\Gamma_2, x : \tau_1 \vdash t_2 : \costmonad{(\kappa_1 + \kappa_2)}{\tau_2}$}
    \RightLabel{\textsc{Release}}
    \BinaryInfC{$\Gamma_1,\Gamma_2 \vdash \tmrelease{x}{t_1}{t_2} : \costmonad{\kappa_2}{\tau_2}$}
\end{prooftree}
\end{minipage}
\vspace{0.1cm}

The operation $\tmstore{t}$ incurs some cost $\kappa$ now, which is reflected in the monad, and
attaches it as potential in the type $\potentialtype{\kappa}{\tau}$ so that it can be used later.
The operation $\tmrelease{x}{t_1}{t_2}$ uses some potential stored with the term $t_1$ to pay for
some of the cost in the monadic computation $t_2$.

To prove that the amortized analysis provided by the cost monad and the potential type is sound we
can use denotational models. \citet{Rajani2021lamor} develop a step-indexed Kripke model (indexed
sets of terms) using it to establish that running a monadic computation never exceeds the cost
promised by the monad. That model is very concrete, which makes it difficult to combine with other
language features that need to be modelled differently. Building a categorical model will help bring
out the necessary structure needed to handle amortization, which can later be combined in a
compositional way with other kinds of deterministic~\cite{Rajani2021lamor} (like recursion or
bounded use) or probabilistic effects~\cite{Rajani2024modal} (like expected loss or regret).
To that end, we make the following contributions in this work:


\begin{description}
  \item[Contribution 1:] We describe an abstract model which models cost and potential by a family of adjunctions between a strong graded monad and graded comonad (for cost) and a strong graded functor (for potential).
     We describe in detail the properties we require to model all features of the type system.
     We prove that the model is sound, i.e. we prove that every typing derivation can be interpreted in the model.
  \item[Contribution 2:] We provide a mechanization of the type system in the Lean proof assistant and verify that the Kripke model presented by \citet{Rajani2021lamor} is indeed an instance
     of the general categorical model that we have identified.
     This has not been observed by \citet{Rajani2021lamor}, and proving that the necessary laws hold involves a nontrivial amount of technical detail.
  \item[Contribution 3:] We present a novel model instance which uses presheaves on a monoidal category of costs.
     We use Day convolution \cite{Day1970thesis} to model positive pairs, and its right adjoint to model the function type.
     Presheaf models are widely used in categorical semantics, and our model shows that they can also be used for cost and potential.
\end{description}

The rest of this article is structured as follows:
\begin{itemize}
    \item In \cref{sec:example} we present one of the classical examples of amortized cost analysis in the $\lambda$-amor style: The two-list representation of queues.
    \item In \cref{sec:lamor} we present the syntax, typing rules and operational semantics of a small subset of $\lambda$-amor which has enough structure to explain the core mechanisms used to track cost in the type system, but which lacks several abstraction mechanisms which are needed to analyze more complicated examples.
    In particular, the system described in that section does not have exponentials, universal and existential quantification, and recursion.
    \item We investigate the mathematical structure of $\lambda$-amor through a discussion of its possible categorical semantics in \cref{sec:lamor-cat-model}.
    This then motivates an alternate, more reduced, calculus, to which we give the name \textbf{$\lambda$-amor Core},
    presented in \cref{sec:lambda-amor-core}. This calculus comes
    with an equational theory which was missing from the original presentation of $\lambda$-amor.
  \item In \cref{sec:adjoint-model} we present a general model for $\lambda$-amor Core using a
    family of adjunctions and prove its soundness with respect to the equational theory introduced
    in \cref{sec:lambda-amor-core}.
  \item In \cref{sec:kripke-model} we show that the Kripke model presented in \cite{Rajani2021lamor}
    is an instance of the generic model from previous section.
    \item In \cref{sec:model} we present another instance based on covariant preshaves (\emph{copresheaves}) over a category of costs which gives a kind of canonical construction for amortized cost models.
\end{itemize}
We compare to related work in \cref{sec:related-work}, discuss future work in \cref{sec:future-work} and conclude in \cref{sec:conclusion}.
Some of the theorems formalized in this paper have been verified in the proof assistant Lean; these theorems are marked with \lean.

\section{Type-Based Amortized Complexity by Example}
\label{sec:example}
Before we delve deeply into the details of the type theory and model, we will first give a high-level introduction to the kind of problems which motivate their design.
In this section we are therefore going to take a look at a classical example of amortized analysis of an algorithm: \citet[Chapter 3]{Okasaki1998purely}'s presentation of a first-in-first-out (FIFO) queue.
We first introduce in \cref{subsec:example:informal} the example using the informal style that is used in books and paper before we then show how to formally analyse it using the type system of $\lambda$-amor in \cref{subsec:example:formal}.\footnote{%
Our presentation of this example closely follows \citet[Section 3.1]{Rajani2021lamor}, who modelled this in $\lambda$-amor.}

\subsection{Amortized Analysis, Informal Style}
\label{subsec:example:informal}

FIFO queues are characterized by the two operations illustrated in \cref{fig:example:fifo-queue}: $\mathtt{enqueue}$ adds another element to the front of the queue, and $\mathtt{dequeue}$ which removes an element from the end.
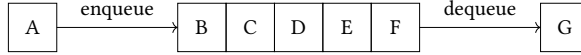
\begin{figure}[h]
  \scalebox{0.8}{
  \begin{tikzpicture}[scale=0.8]
    \draw (-0.5,0) rectangle ++(1,1) node[pos=.5] {A};
    \draw[->] (0.5,0.5) -- (3,0.5) node[midway,above] {enqueue};
    \draw (3,0) rectangle ++(1,1) node[pos=.5] {B};
    \draw (4,0) rectangle ++(1,1) node[pos=.5] {C};
    \draw (5,0) rectangle ++(1,1) node[pos=.5] {D};
    \draw (6,0) rectangle ++(1,1) node[pos=.5] {E};
    \draw (7,0) rectangle ++(1,1) node[pos=.5] {F};
    \draw[->] (8,0.5) -- (10.5,0.5) node[midway,above] {dequeue};
    \draw (10.5,0) rectangle ++(1,1) node[pos=.5] {G};
  \end{tikzpicture}}
  \caption{The operations provided by the interface of a FIFO queue.}
  \Description{The operations provided by the interface of a FIFO queue.}
  \label{fig:example:fifo-queue}
\end{figure}

We are interested in the algorithmic complexity of $\mathtt{enqueue}$ and $\mathtt{dequeue}$, and specifically want to obtain an upper-bound for the number of \enquote{cons} operations that they use.
If the queue is na\"{i}vely represented using a singly-linked list (like those available in many functional programming languages), then $\mathtt{enqueue}$ is always an $O(1)$ operation, since we only need to cons one element to the front of the list, and $\mathtt{dequeue}$ is always an $O(n)$ operation, since we always need to construct a new list with one element less at the \emph{end}.
Luckily, we can do better than that.

One optimized implementation that is widely used in practice represents a FIFO queue by \emph{two} singly-linked lists, where the first part of the queue is stored in the \emph{front list}, and the second part of the queue is stored in reverse order in the \emph{back list}.
This is illustrated in \cref{fig:example:two-list-repr}, which shows the effect of enqueuing and dequeing elements to the queue.

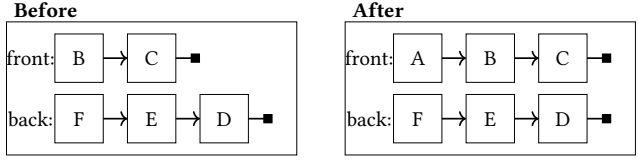
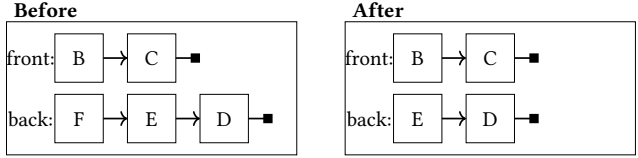
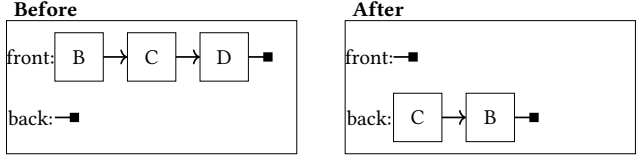
\begin{figure}[h]
  \begin{subfigure}{\textwidth}
    \begin{center}
    \scalebox{0.8}{
    \begin{tikzpicture}[scale=0.8]
        \node[right] at (0,3.25) {\textbf{Before}};
        \draw (0,0.25) rectangle ++(6,2.75);
        \node at (0.5,2.25) {front:};
        \draw (1,1.75) rectangle ++(1,1) node[pos=.5] {B};
        \draw (2.5,1.75) rectangle ++(1,1) node[pos=.5] {C};
        \draw[->, thick] (2,2.25) -- ++(0.5,0);
        \draw[-{Square}, thick] (3.5,2.25) -- ++(0.5,0);
        \node at (0.5,1) {back:};
        \draw (1,0.5) rectangle ++(1,1) node[pos=.5] {F};
        \draw (2.5,0.5) rectangle ++(1,1) node[pos=.5] {E};
        \draw (4,0.5) rectangle ++(1,1) node[pos=.5] {D};
        \draw[->, thick] (2,1) -- ++(0.5,0);
        \draw[->, thick] (3.5,1) -- ++(0.5,0);
        \draw[-{Square}, thick] (5,1) -- ++(0.5,0);

        \node[right] at (7,3.25) {\textbf{After}};
        \draw (7,0.25) rectangle ++(6,2.75);
        \node at (7.5,2.25) {front:};
        \draw (8,1.75) rectangle ++(1,1) node[pos=.5] {A};
        \draw (9.5,1.75) rectangle ++(1,1) node[pos=.5] {B};
        \draw (11,1.75) rectangle ++(1,1) node[pos=.5] {C};
        \draw[->, thick] (9,2.25) -- ++(0.5,0);
        \draw[->, thick] (10.5,2.25) -- ++(0.5,0);
        \draw[-{Square}, thick] (12,2.25) -- ++(0.5,0);
        \node at (7.5,1) {back:};
        \draw (8,0.5) rectangle ++(1,1) node[pos=.5] {F};
        \draw (9.5,0.5) rectangle ++(1,1) node[pos=.5] {E};
        \draw (11,0.5) rectangle ++(1,1) node[pos=.5] {D};
        \draw[->, thick] (9,1) -- ++(0.5,0);
        \draw[->, thick] (10.5,1) -- ++(0.5,0);
        \draw[-{Square}, thick] (12,1) -- ++(0.5,0);
    \end{tikzpicture}}
    \end{center}
    \caption{Before and after enqueuing A; uses at most $O(1)$ cons operations.}
    \label{fig:example:two-list-repr:enqueue}
  \end{subfigure}

  \begin{subfigure}{\textwidth}
    \begin{center}
    \scalebox{0.8}{
    \begin{tikzpicture}[scale=0.8]
        \node[right] at (0,3.25) {\textbf{Before}};
        \draw (0,0.25) rectangle ++(6,2.75);
        \node at (0.5,2.25) {front:};
        \draw (1,1.75) rectangle ++(1,1) node[pos=.5] {B};
        \draw (2.5,1.75) rectangle ++(1,1) node[pos=.5] {C};
        \draw[->, thick] (2,2.25) -- ++(0.5,0);
        \draw[-{Square}, thick] (3.5,2.25) -- ++(0.5,0);
        \node at (0.5,1) {back:};
        \draw (1,0.5) rectangle ++(1,1) node[pos=.5] {F};
        \draw (2.5,0.5) rectangle ++(1,1) node[pos=.5] {E};
        \draw (4,0.5) rectangle ++(1,1) node[pos=.5] {D};
        \draw[->, thick] (2,1) -- ++(0.5,0);
        \draw[->, thick] (3.5,1) -- ++(0.5,0);
        \draw[-{Square}, thick] (5,1) -- ++(0.5,0);

        \node[right] at (7,3.25) {\textbf{After}};
        \draw (7,0.25) rectangle ++(6,2.75);
        \node at (7.5,2.25) {front:};
        \draw (8,1.75) rectangle ++(1,1) node[pos=.5] {B};
        \draw (9.5,1.75) rectangle ++(1,1) node[pos=.5] {C};
        \draw[->, thick] (9,2.25) -- ++(0.5,0);
        \draw[-{Square}, thick] (10.5,2.25) -- ++(0.5,0);
        \node at (7.5,1) {back:};
        \draw (8,0.5) rectangle ++(1,1) node[pos=.5] {E};
        \draw (9.5,0.5) rectangle ++(1,1) node[pos=.5] {D};
        \draw[->, thick] (9,1) -- ++(0.5,0);
        \draw[-{Square}, thick] (10.5,1) -- ++(0.5,0);
    \end{tikzpicture}}
    \end{center}
    \caption{Before and after dequeuing F with non-empty back list; uses at most $O(1)$ cons operations.}
    \label{fig:example:two-list-repr:dequeue-simple}
  \end{subfigure}

  \begin{subfigure}{\textwidth}
    \begin{center}
    \scalebox{0.8}{
    \begin{tikzpicture}[scale=0.8]
        \node[right] at (0,3.25) {\textbf{Before}};
        \draw (0,0.25) rectangle ++(6,2.75);
        \node at (0.5,2.25) {front:};
        \draw (1,1.75) rectangle ++(1,1) node[pos=.5] {B};
        \draw (2.5,1.75) rectangle ++(1,1) node[pos=.5] {C};
        \draw (4,1.75) rectangle ++(1,1) node[pos=.5] {D};
        \draw[->, thick] (2,2.25) -- ++(0.5,0);
        \draw[->, thick] (3.5,2.25) -- ++(0.5,0);
        \draw[-{Square}, thick] (5,2.25) -- ++(0.5,0);
        \node at (0.5,1) {back:};
        \draw[-{Square}, thick] (1,1) -- ++(0.5,0);

        \node[right] at (7,3.25) {\textbf{After}};
        \draw (7,0.25) rectangle ++(6,2.75);
        \node at (7.5,2.25) {front:};
        \draw[-{Square}, thick] (8,2.25) -- ++(0.5,0);
        \node at (7.5,1) {back:};
        \draw (8,0.5) rectangle ++(1,1) node[pos=.5] {C};
        \draw (9.5,0.5) rectangle ++(1,1) node[pos=.5] {B};
        \draw[->, thick] (9,1) -- ++(0.5,0);
        \draw[-{Square}, thick] (10.5,1) -- ++(0.5,0);
    \end{tikzpicture}}
    \end{center}
    \caption{Before and after dequeuing F with empty back list; uses $O(n)$ cons operations, with $n =$ length of front list.}
    \label{fig:example:two-list-repr:dequeue-complicated}
  \end{subfigure}

  \caption{Enqueue and dequeue for the two-list representation of FIFO queues.}
  \label{fig:example:two-list-repr}
\end{figure}

In \cref{fig:example:two-list-repr:enqueue} we can see that enqueuing an element $A$ to the front of the list requires exactly one cons operation, and therefore is in $O(1)$.
\cref{fig:example:two-list-repr:dequeue-simple} and \cref{fig:example:two-list-repr:dequeue-complicated} illustrate the dequeue operation:
If the back list is non-empty, as in \cref{fig:example:two-list-repr:dequeue-simple}, then we do not have to perform any cons-ing, and the operation is clearly in $O(1)$.
If the back list is empty however, as in \cref{fig:example:two-list-repr:dequeue-complicated}, we need to perform $n-1$ cons operations to reverse the front list and create the new back list.
This analysis shows that the \emph{worst-case} of these operations is $O(n)$.

What we want to establish instead, however, is the following \emph{amortized} guarantee:
Every sequence of $n$ $\mathtt{enqueue}$ and $\mathtt{dequeue}$ operations uses the cons constructor at most $O(n)$ times.
We can establish this property using the so-called \enquote{Banker's method} \cite{Tarjan1985amortized,Okasaki1998purely}.
Every element of the front list is associated with some potential.\footnote{Also sometimes called \enquote{credits}, but we use the term \enquote{potential} throughout.}
This potential is only used for bookkeeping and not present at runtime.
Whenever we enqueue an element to the front list we now incur a cost of 2, one for the cons operation and one for storing some potential in the newly cons-ed element.
The payoff comes when we consider the dequeue case of \cref{fig:example:two-list-repr:dequeue-complicated} above:
We can now use the $n$ units of stored potential in the front list to pay for the construction of the reversed back list.

Note that this analysis is only sound under one of two restrictions:
Either queues are only used affinely, i.e. we never invoke $\mathtt{dequeue}$ multiple times on the same queue, or we use lazyness to ensure that the result of the expensive operation is shared and never computed more than once.
Lazyness makes this sort of theoretical analysis sound for implementations in call-by-need languages like Haskell, but for the purposes of this paper, and following \citet{Rajani2021lamor}, we enforce affine use of resources using a substructural type system.

\subsection{Type-Based Amortized Analysis}
\label{subsec:example:formal}

Let us now see how we can express the informal analysis of the previous section using a type system, such as the one presented by \citet{Rajani2021lamor}.
We use some features present in the type system of full $\lambda$-amor but not our formalization (indexed types for sized vectors) and omit terms.

We start with a definition of the $\mathtt{cons}$ function which wraps the constructor $::$ with a call to the tick function, representing the fact that we want to account for the number of $::$ constructors used.
\begin{align*}
  &\mathtt{cons} : \forall a\ n, a \multimap \tyvec{n}{a} \multimap \costmonad{1}{(\tyvec{(n + 1)}{a})}
\end{align*}
Using this function we can also define a function $\mathtt{reverse}$ whose cost is bound by the length of the original vector:
Reversing a vector of length $n$ incurs a cost of $n$.
\begin{align*}
  &\mathtt{reverse} : \forall a\ n, \tyvec{n}{a} \multimap \costmonad{n}{(\tyvec{n}{a})}
\end{align*}
By storing 1 potential with each element of the vector, i.e. by replacing the type $\tyvec{n}{a}$ with the type $\tyvec{n}{(\potentialtype{1}{a})}$, we can change the cost of $\mathtt{cons}$ and $\mathtt{reverse}$ and implement the following two functions instead:
\begin{align*}
  &\mathtt{cons}' : \forall a\ n, a \multimap \tyvec{n}{(\potentialtype{1}{a})} \multimap \costmonad{2}{(\tyvec{(n + 1)}{(\potentialtype{1}{a})})}\\
  &\mathtt{reverse}' : \forall a\ n, \tyvec{n}{(\potentialtype{1}{a})} \multimap \costmonad{0}{(\tyvec{n}{a})}
\end{align*}
We can see that every $\mathtt{cons'}$ operation now incurs a cost of $2$ instead of $1$, but reversing the list using $\mathtt{reverse}'$ now comes for free by consuming the potential stored in the vector.
This motivates the following definition of queues of $n$ elements of type $a$:
\begin{align*}
  &\tyqueue{n}{a} \coloneq \exists m\ p, m + p = n \wedge  \tyvec{m}{(\potentialtype{1}{a})} \otimes \tyvec{p}{a}
\end{align*}
A queue of length $n = m + p$ stores $m$ elements in the front list together with some potential, and $p$ elements in the back list without potential.
Using the $\mathtt{cons}'$ and $\mathtt{reverse'}$ functions from above we can now implement the $\mathtt{enqueue}$ and $\mathtt{dequeue}$ operations with the following type signatures.
\begin{align*}
  &\mathtt{enqueue} : \forall n\ a, a \multimap \tyqueue{n}{a} \multimap \costmonad{2}{(\tyqueue{(n+1)}{a})}\\
  &\mathtt{dequeue} : \forall n\ a, \tyqueue{(n + 1)}{a} \multimap \costmonad{0}{(a \otimes \tyqueue{n}{a})}
\end{align*}
Starting with an empty queue, and assuming that the type system is sound, we can observe that running a computation which chains $n$ $\mathtt{enqueue}$ and $\mathtt{dequeue}$ operations will run tick $\tmtick{1}$ at most $2n$ times.
This formally proves that these operations have amortized complexity of $O(1)$.

\section{Lambda-Amor}
\label{sec:lamor}
The original presentation of $\lambda$-amor \cite{Rajani2021lamor} chose a specific domain to model the cost of computations: positive real numbers $\mathbb{R}^+$.
In this article, we want to generalize their approach and extract the essential abstract properties that a domain for costs has to fulfil.
We therefore generalize costs from $\mathbb{R}^+$ to an arbitrary ordered commutative monoid $\costcategory$.

\begin{definition}[Ordered Commutative Monoid]
  \label{def:lamor:ordered-commutative-monoid}
  An ordered commutative monoid $(\costcategory, 0, +, \leq)$ combines a commutative monoid $(\costcategory, 0, +)$ and a pre-order $(\costcategory, \leq)$.
  Furthermore, the monoid operation has to be compatible with the pre-order.
  This means that whenever $x \leq y$, then $x + z \leq y + z$.
\end{definition}

The notation for the monoid operation and unit has been chosen deliberately: The standard intended model for costs are positive natural or real numbers with 0 and addition.
We also require all cost models to have a least element which coincides with the unit $0$.

\begin{example}[Bounded Cost]
  \label{ex:lamor:bounded-cost}
  The ordered monoid $(\mathbb{N}, 0, +, \leq)$ with the usual interpretation of $0$, $+$ and $\leq$ can be used to model bounded cost.
\end{example}

The following definition of the syntax of $\lambda$-amor is then parameterized over the choice of a specific ordered monoid for cost.
This parameterization can be seen in the syntax of indices, where constants $c$ are used to inject elements of the ordered monoid into the syntax.

In this article, we are not formalising the entire language presented in \citet{Rajani2021lamor} but only the subset without exponentials, since this is sufficient to show the model constructions discussed here.
There are some minor syntactic
variations introduced for the sake of notational consistency and one more
deviation with the addition of the $\mathtt{run}$ construct which appears in
the later extension of p$\lambda$-amor~\cite{Rajani2024modal}, which we explain in detail below.

\begin{definition}[Syntax of $\lambda$-amor core]
  \[
  \begin{array}{lclr}
    \tau & \Coloneqq & \tyunitp \mid \tyunitn \mid \tyvoid \mid \tau \multimap \tau \mid \tau \otimes \tau \mid \tau \with \tau \mid \tau \oplus \tau \mid \potentialtype{p}{\tau} \mid \costmonad{p}{\tau} &\emph{Types} \\
    e & \Coloneqq & v \mid x \mid e_1\ e_2 \mid \tmcase{e}{\tminl{x} \Rightarrow e, \tminr{x} \Rightarrow e} \mid e.\tmfst \mid e.\tmsnd & \emph{Expressions} \\
    & \mid & \tmcase{e}{\tmtensor{x}{y} \Rightarrow e} \mid \tmcase{e}{\tmunit \Rightarrow e} \mid \tmcase{e}{} \mid \tmrun{e} & \\
    v & \Coloneqq & \lambda x.e \mid \tmunit \mid \tminl{e} \mid \tminr{e} \mid \tmtensor{e_1}{e_2} \mid \tmcocase{\tmfst \Rightarrow e, \tmsnd \Rightarrow e} & \emph{Values} \\
    & \mid & \tmcocase{} \mid \tmret{e} \mid \tmtick{p} \mid \tmbind{x}{e}{e} \mid \tmstore{e} \mid \tmrelease{x}{e}{e} & \\
    p & \Coloneqq & 0 \mid p + p \mid c & \emph{Indices} \\
    \Gamma & \Coloneqq & \cdot \mid \Gamma, x : \tau & \emph{Typing contexts}
  \end{array}
  \]
\end{definition}

Let us start with the ordinary logical connectives that $\lambda$-amor core has in common with many other substructural type systems.
Affine functions $\tau_1 \multimap \tau_2$ are introduced by lambda abstraction $\lambda x.e$ and eliminated by function application $e_1\, e_2$.
There are two types of products: multiplicative products $\tau_1 \otimes \tau_2$ (pronounced \enquote{tensor}) are introduced by pairing two elements with $\tmtensor{e_1}{e_2}$ and eliminated using pattern matching $\tmcase{e_1}{\tmtensor{x}{y} \Rightarrow e_2}$, whereas additive products $\tau_1 \with \tau_2$ (pronounced \enquote{with}) are introduced using copattern matching $\tmcocase{\tmfst \Rightarrow e_1, \tmsnd \Rightarrow e_2}$ and eliminated using projections $e.\tmfst$ and $e.\tmsnd$.
The unit for $\otimes$ is $\tyunitp$, which is introduced by $\tmunit$ and eliminated using $\tmcase{e}{\tmunit \Rightarrow e}$, whereas the unit for $\with$ is $\tyunitn$, which is introduced by an empty $\tmcocase{}$ and has no elimination rule.
Sums $\tau_1 \oplus \tau_2$ are introduced using left injections $\tminl{e}$ and right injections $\tminr{e}$, and eliminated using a pattern match on both alternatives $\tmcase{e_1}{\tminl{x} \Rightarrow e_2, \tminr{y} \Rightarrow e_3}$.
The unit for $\oplus$ is $\tyvoid$ which has the elimination rule $\tmcase{e}{}$ and no introduction rule.
The types $\otimes$, $\oplus$, $\tyunitp$, $\tyvoid$ are positive, or data types, whereas $\with$, $\tyunitn$ and $\multimap$ are negative, or codata types \cite{Downen2019codata}.

The type constructor $\costmonad{p}{\tau}$ describes computations which incur
a cost of at most $p \in \mathbb{C}$ in order to produce a value of type $\tau$.
This type constructor has the structure of a \emph{graded monad}~\cite{Katsumata2014parametric} and \emph{graded comonad} and comes with three constructs: $\mathtt{bind}$ for sequencing computations,
$\mathtt{ret}$ for injecting pure computations into computations, and $\tmtick{p}$
for incurring a cost of $p$.

Dual to the idea of incurring cost (`consuming fuel') is that of
storing potential (`producing fuel'). Potentials are modelled using
the type $\potentialtype{p}{\tau}$ denoting a computation which
computes a value of type $\tau$ whilst also storing potential
$p \in \mathbb{C}$ for later use.  Two constructs $\mathtt{store}$ and
$\mathtt{release}$ provide the interaction between potential and cost.
We defer the rest of the explanation of the meaning of the syntax
until the typing relation is in view, since the meaning of constructs
for cost and potential tracking can be explained more clearly through
their types.

\subsection{Typing Rules}
\label{subsec:lamor:typing-rules}
We write $\mathcal{D} : \Gamma \vdash e : \tau$ to say that the typing derivation $\mathcal{D}$ proves the judgement $\Gamma \vdash e : \tau$.
In \cref{sec:lambda-amor-core} we are going to develop an equational theory on typing derivations.
For example, given the following rule we will write $\textsc{Tensor}(\mathcal{D}_1, \mathcal{D}_2)$ to denote the derivation that results from combining the two subderivations.

\begin{prooftree}
    \AxiomC{$\Gamma_1 \vdash e_1 : \tau_1$}
    \AxiomC{$\Gamma_2 \vdash e_2 : \tau_2$}
    \RightLabel{\textsc{Tensor}}
    \BinaryInfC{$\Gamma_1, \Gamma_2 \vdash \tmtensor{e_1}{e_2} : \tau_1 \otimes \tau_2$}
\end{prooftree}

Typing derivations are defined by three classes of rules: Core structural rules (\cref{def:lamor:core-typing-rules}), logical rules (\cref{def:lamor:logical-typing-rules}), and rules for the cost monad and potential type (\cref{def:lamor:typing-cost-potential}).

\begin{definition}[Core Typing Rules]
  \label{def:lamor:core-typing-rules}
  These structural rules are independent of any specific types that are present in the language.

  \begin{minipage}{0.45\textwidth}
    \tyrulevar
  \end{minipage}
  \begin{minipage}{0.45\textwidth}
    \tyrulesub
  \end{minipage}\\

  \begin{minipage}{0.45\textwidth}
    \tyrulewk
  \end{minipage}
  \begin{minipage}{0.45\textwidth}
    \tyruleex
  \end{minipage}
\end{definition}

The system $\lambda$-amor is affine, so \cref{def:lamor:core-typing-rules} only contains rules $\textsc{Wk}$ for weakening and $\textsc{Ex}$ for exchange, but no rule
for contraction.
The rule $\textsc{Sub}$ formalizes subsumption both in the type of the term $e$ and in the context that we are using to type $e$, leveraging the subtyping
relation (Section~\ref{subsec:lamor:subtyping}).

\begin{definition}[Typing Rules for Logical Constants]
  \label{def:lamor:logical-typing-rules}
  The following rules govern the logical connectives present in the system:

  \begin{minipage}{0.45\textwidth}
    \tyruleabs
  \end{minipage}
  \begin{minipage}{0.45\textwidth}
    \tyruleapp
  \end{minipage}

  \begin{minipage}{0.3\textwidth}
    \tyruleinl
  \end{minipage}
  \begin{minipage}{0.3\textwidth}
    \tyruleinr
  \end{minipage}
  \begin{minipage}{0.3\textwidth}
    \tyrulevoid
  \end{minipage}

  \begin{minipage}{0.35\textwidth}
    \tyruleunitp
  \end{minipage}
  \begin{minipage}{0.5\textwidth}
    \tyrulecaseunitp
  \end{minipage}

  \tyrulecasesum

  \tyruletensor

  \begin{minipage}{0.6\textwidth}
    \tyrulecasetensor
  \end{minipage}
  \begin{minipage}{0.35\textwidth}
    \tyruleunitn
  \end{minipage}

  \begin{minipage}{0.52\textwidth}
  \tyrulewith
  \end{minipage}
  \begin{minipage}{0.22\textwidth}
    \tyrulefst
  \end{minipage}
  \begin{minipage}{0.22\textwidth}
    \tyrulesnd
  \end{minipage}
\end{definition}

\begin{definition}[Typing Rules for Cost and Potential]
  \label{def:lamor:typing-cost-potential}
  The following typing rules govern the operations for the cost graded monad $\costmonad{p}{\tau}$ and potentials $\potentialtype{p}{\tau}$, and their interaction.

  \vspace{0.1cm}
  \begin{minipage}{0.6\textwidth}
  \tyrulebind
  \end{minipage}
  \begin{minipage}{0.4\textwidth}
    \tyruleret
  \end{minipage} \\

  \begin{minipage}{0.5\textwidth}
    \tyrulerun
  \end{minipage}
  \begin{minipage}{0.5\textwidth}
    \tyruletick
  \end{minipage}
  \vspace{0.1cm}


The \textsc{Bind} and \textsc{Ret} rules are standard for graded monads.
The expression $\tmbind{x}{e_1}{e_2}$ captures sequential composition of
a cost-incurring computation $e_1$, bindings its result to $x$ in the scope of $e_2$, leading to the additive composition of costs. The expression $\tmret{e}$ allows a pure term to be `injected' into the graded monad, marked as having $0$ cost.
The \textsc{Tick} rule types an effectful operation denoted $\tmtick{p}$ incurring cost $p$. 

The \textsc{Run} typing rule captures that we can project out of a costful
computation when the cost is known to be $0$; this witnesses that all
cost-incurring effects that are possible in the graded monad are
tracked by the grade, and thus a computation at grade 0 does not use
any effects.
This primitive is where we deviate slightly from the original definition of $\lambda$-amor~\cite{Rajani2021lamor} which
had no such primitive.
However, a later extension of $\lambda$-amor to probabilistic programs, called p$\lambda$-amor~\cite{Rajani2024modal}, provides
in its technical appendix~\cite{Rajani2024technical}
a `purify' primitive with the same typing as \textsc{Run}, although extended
to the indexed-type setting of p$\lambda$-amor.
We include this primitive here and discuss how it is important for our development in Section~\ref{sec:lambda-amor-core}.

Potential is then introduced and eliminated in typing rules that specify an interaction with cost:

  \begin{minipage}{0.35\textwidth}
    \tyrulestore
  \end{minipage}
  \begin{minipage}{0.65\textwidth}
    \tyrulerelease
  \end{minipage} \\

Dual to the idea of incurring cost is that of \emph{storing} potential, which is triggered by the
$\tmstore{e}$ constructor. Its typing (rule \textsc{Store}) shows that we can lift a pure term $e$
into a computation that incurrs $p$ cost to store $p$ potential with the term.

The \textsc{Release} rule explains that, given a term $e_1$ storing a potential of $p_1$ units, we
can use this to account for the cost of a computation $e_2$ that costs $p_1 + p_2$ leaving a cost of
$p_2$ remaining. As part of this, we get a binding for $x : \tau_1$ in the scope of $e_2$. The rule
for $\mathtt{release}$ thus packs together a substitution, an interaction between potential and
cost, and a notion of cost splitting.
\end{definition}

The following definition is often presented as the \enquote{substitution lemma} for typing derivations.
We present it here as a definition instead so that we can use it in the equational theory for typing derivations, where we need it to express $\beta$-equality.

\begin{definition}[Substitution on Typing Derivations]
  \label{def:lamor:substitution-derivations}
  We will formulate this as two derived operations on typing operations, once for the unary and once for the binary case.
  \begin{prooftree}
    \AxiomC{$\Gamma_1, x : \tau_1 \vdash e_1 : \tau_2$}
    \AxiomC{$\Gamma_2 \vdash e_2 : \tau_1$}
    \RightLabel{\textsc{Subst}$_1$}
    \BinaryInfC{$\Gamma_1, \Gamma_2 \vdash e_1[e_2/x] : \tau_2$}
  \end{prooftree}
  \begin{prooftree}
    \AxiomC{$\Gamma_1, x : \tau_1, y : \tau_2 \vdash e_1 : \tau_3$}
    \AxiomC{$\Gamma_2 \vdash e_2 : \tau_1$}
    \AxiomC{$\Gamma_3 \vdash e_3 : \tau_2$}
    \RightLabel{\textsc{Subst}$_2$}
    \TrinaryInfC{$\Gamma_1, \Gamma_2, \Gamma_3 \vdash e_1[e_2/x, e_3/y] : \tau_3$}
  \end{prooftree}
  In both cases the simultaneous definition and proof proceeds by recursion on the structure of the typing derivation of $e_1$.
\end{definition}

\subsection{Subtyping}
\label{subsec:lamor:subtyping}

The costs tracked by our system are ordered.
Since these costs can also appear in types, such as $\potentialtype{p}{\tau}$ or $\costmonad{p}{\tau}$, we obtain a subtyping order for types.
This subtyping order is only induced by the order on costs, we do not consider any other subtyping relationship based on type structure.

\begin{definition}[Subtyping]
  The subtyping relationship is characterized by the following rules, where we have omitted the usual congruence rules for type formers.
  \vspace{0.2cm}

  \begin{minipage}{0.3\textwidth}
    \begin{prooftree}
      \AxiomC{}
      \RightLabel{$\textsc{Refl-}\tyunitp$}
      \UnaryInfC{$\tyunitp \sub \tyunitp$}
    \end{prooftree}
  \end{minipage}
  \begin{minipage}{0.3\textwidth}
    \begin{prooftree}
      \AxiomC{}
      \RightLabel{$\textsc{Refl-}\tyunitn$}
      \UnaryInfC{$\tyunitn \sub \tyunitn$}
    \end{prooftree}
  \end{minipage}
  \begin{minipage}{0.3\textwidth}
    \begin{prooftree}
      \AxiomC{}
      \RightLabel{$\textsc{Refl-}\tyvoid$}
      \UnaryInfC{$\tyvoid \sub \tyvoid$}
    \end{prooftree}
  \end{minipage}
  \vspace{0.15cm}



  \begin{minipage}{0.45\textwidth}
    \begin{prooftree}
      \AxiomC{$p_2 \leq p_1$}
      \AxiomC{$\tau_1 \sub \tau_2$}
      \RightLabel{\textsc{S-Pot}}
      \BinaryInfC{$\potentialtype{p_1}{\tau_1} \sub \potentialtype{p_2}{\tau_2}$}
    \end{prooftree}
  \end{minipage}
  \begin{minipage}{0.45\textwidth}
    \begin{prooftree}
      \AxiomC{$p_1 \leq p_2$}
      \AxiomC{$\tau_1 \sub \tau_2$}
      \RightLabel{\textsc{S-Mon}}
      \BinaryInfC{$\costmonad{p_1}{\tau_1} \sub \costmonad{p_2}{\tau_2}$}
    \end{prooftree}
  \end{minipage}
  \vspace{0.15cm}



  \label{def:lamor:subtyping-derivation}
\end{definition}

We write $\mathcal{S} : \tau_1 \sub \tau_2$ to express that $\mathcal{S}$ is a derivation which shows that $\tau_1$ is a subtype of $\tau_2$.
Subtyping between types can be extended in the obvious way to subtyping between contexts, and we write $\mathcal{X} : \Gamma_1 \sub \Gamma_2$ for a derivation $\mathcal{X}$ which proves that $\Gamma_1$ is a subcontext of $\Gamma_2$.

We do not have general transitivity or reflexivity rules, but three different reflexivity rules for the atomic types $\tyunitp, \tyunitn$ and $\tyvoid$ instead.
However, general reflexivity can be derived and transitivity is admissible, as the following two lemmas show.
\begin{lemma}[Derivability of Reflexivity, \lean]
  For any type $\tau$ and context $\Gamma$, there are derivations $\textsc{Refl} : \tau \sub \tau$ and $\textsc{Refl} : \Gamma \sub \Gamma$.
\end{lemma}
\begin{lemma}[Admissibility of Transitivity, \lean]
    For all types $\tau_1, \tau_2$ and $\tau_3$, if there are derivations $\mathcal{S}_1 : \tau_1 \sub \tau_2$ and $\mathcal{S}_2 : \tau_2 \sub \tau_3$, then there is a derivation $\textsc{Trans}(\mathcal{S}_1,\mathcal{S}_2) : \tau_1 \sub \tau_3$.
    \label{lem:lamor:subtyping-trans}
\end{lemma}




\subsection{Operational Semantics}
\label{subsec:lamor:operational-semantics}

We define two kinds of evaluation relations a pure ($\evaluatesto$) and a forcing
($\forcesto{\cost}$). The pure evaluation (Definition~\ref{def:pure-eval}) is standard and do not
track any cost. Monadic terms ($\mathtt{ret}$, $\tmtick{p}$, $\mathtt{store}$, $\mathtt{bind}$,
$\mathtt{release}$) are treated as values in the pure relation. To evaluate monadic (cost bearing)
computations we use the forcing relation (Definition~\ref{def:forcing-eval}). The relation evaluates
a monadic value of type $\costmonad{p'}{\tau}$ to a value of type $\tau$ while tracking the cost of
such an evaluation. The rules for the both the relations are described below.



\begin{definition}[Pure evaluation]
  \label{def:pure-eval}
  The pure evaluation is described by the following rules:

\begin{minipage}{0.3\textwidth}
  \begin{prooftree}
    \AxiomC{\phantom{X}}
    \RightLabel{\textsc{E-Val}}
    \UnaryInfC{$v \evaluatesto v$}
  \end{prooftree}
\end{minipage}
\begin{minipage}{0.4\textwidth}
  \begin{prooftree}
    \AxiomC{$e_1 \evaluatesto \lambda x.e_3$}
    \AxiomC{$e_3[e_2/x] \evaluatesto v$}
    \RightLabel{\textsc{E-Fun}}
    \BinaryInfC{$e_1\ e_2 \evaluatesto v$}
  \end{prooftree}
\end{minipage}
\begin{minipage}{0.3\textwidth}
 \begin{prooftree}
   \AxiomC{$e \forcesto{0} v$}
   \RightLabel{\textsc{E-Run}}
   \UnaryInfC{$\tmrun{e} \evaluatesto v$}
 \end{prooftree}
\end{minipage}

\begin{prooftree}
  \AxiomC{$e \evaluatesto \tmcocase{\tmfst \Rightarrow e_1, \tmsnd \Rightarrow e_2}$}
  \AxiomC{$e_1 \evaluatesto v$}
  \RightLabel{\textsc{E-With}$_1$}
  \BinaryInfC{$e.\tmfst \evaluatesto v$}
\end{prooftree}

\begin{prooftree}
  \AxiomC{$e \evaluatesto \tmcocase{\tmfst \Rightarrow e_1, \tmsnd \Rightarrow e_2}$}
  \AxiomC{$e_2 \evaluatesto v$}
  \RightLabel{\textsc{E-With}$_2$}
  \BinaryInfC{$e.\tmsnd \evaluatesto v$}
\end{prooftree}

\begin{minipage}{0.5\textwidth}
  \begin{prooftree}
    \AxiomC{$e_1 \evaluatesto \tmtensor{e_2}{e_3}$}
    \AxiomC{$e_4[e_2/x, e_3/y] \evaluatesto v$}
    \RightLabel{\textsc{E-Tensor}}
    \BinaryInfC{$\tmcase{e_1}{\tmtensor{x}{y} \Rightarrow e_4} \evaluatesto v$}
  \end{prooftree}
\end{minipage}
\begin{minipage}{0.4\textwidth}
  \begin{prooftree}
    \AxiomC{$e_1 \evaluatesto \tmunit$}
    \AxiomC{$e_2 \evaluatesto v$}
    \RightLabel{\textsc{E-Unit}}
    \BinaryInfC{$\tmcase{e_1}{\tmunit \Rightarrow e_2} \evaluatesto v$}
  \end{prooftree}
\end{minipage}

\begin{prooftree}
  \AxiomC{$e_1 \evaluatesto \tminl{e_4}$}
  \AxiomC{$e_2[e_4/x] \evaluatesto v$}
  \RightLabel{\textsc{E-Plus}$_1$}
  \BinaryInfC{$\tmcase{e_1}{\tminl{x} \Rightarrow e_2, \tminr{y} \Rightarrow e_3} \evaluatesto v$}
\end{prooftree}

\begin{prooftree}
  \AxiomC{$e_1 \evaluatesto \tminr{e_4}$}
  \AxiomC{$e_3[e_4/y] \evaluatesto v$}
  \RightLabel{\textsc{E-Plus}$_2$}
  \BinaryInfC{$\tmcase{e_1}{\tminl{x} \Rightarrow e_2, \tminr{y} \Rightarrow e_3} \evaluatesto v$}
\end{prooftree}

\end{definition}


\begin{definition}[Forcing evaluation]
  \label{def:forcing-eval}
The forcing relation is described by the following rules:

\begin{minipage}{0.3\textwidth}
  \begin{prooftree}
    \AxiomC{$e \evaluatesto v$}
    \RightLabel{\textsc{F-Return}}
    \UnaryInfC{$\tmret{e} \forcesto{0} v$}
  \end{prooftree}
\end{minipage}
\begin{minipage}{0.6\textwidth}
\begin{prooftree}
  \AxiomC{$e_1 \evaluatesto v_1$}
  \AxiomC{$v_1 \forcesto{p_1} v_1'$}
  \AxiomC{$e_2[v_1'/x] \evaluatesto v_2$}
  \AxiomC{$v_2 \forcesto{p_2} v_2'$}
  \RightLabel{\textsc{F-Bind}}
  \QuaternaryInfC{$\tmbind{x}{e_1}{e_2} \forcesto{p_1 + p_2} v_2'$}
\end{prooftree}
\end{minipage}

\hspace{-2em}\begin{minipage}{0.24\textwidth}
  \begin{prooftree}
    \AxiomC{$\phantom{e \evaluatesto v}$}
    \RightLabel{\textsc{F-Tick}}
    \UnaryInfC{$\tmtick{p} \forcesto{p} \tmunit$}
  \end{prooftree}
\end{minipage}
\begin{minipage}{0.2\textwidth}
  \begin{prooftree}
    \AxiomC{$e \evaluatesto v$}
    \RightLabel{\textsc{F-Store}}
    \UnaryInfC{$\tmstore{e} \forcesto{0} v$}
  \end{prooftree}
\end{minipage}
\begin{minipage}{0.6\textwidth}
\begin{prooftree}
  \AxiomC{$e_1 \evaluatesto v_1$}
  \AxiomC{$e_2[v_1/x] \evaluatesto v_2$}
  \AxiomC{$v_2 \forcesto{p} v_2'$}
  \RightLabel{\textsc{F-Release}}
  \TrinaryInfC{$\tmrelease{x}{e_1}{e_2} \forcesto{p} v_2'$}
\end{prooftree}
\end{minipage}
\end{definition}


\section{Considering the Categorical Semantics of Lambda-Amor}
\label{sec:lamor-cat-model}

Our overall goal is a category theoretic account of the semantics of
amortised cost analysis, taking $\lambda$-amor as our starting point.
An investigation of a semantics for $\lambda$-amor, leveraging
standard principles, leads us to an alternate design for the calculus,
based on a more primitive set of constructs arising from the categorical
model. Since we would like to be able to study how to combine cost and
potential with other notions of computation, this reduced set of primitives
provides a more useful substrate for future extensions.

We unpack the categorical analysis of the semantics
of $\lambda$-amor here informally, without going into full definitions,
before proposing the refined calculus of
\textbf{$\lambda$-amor Core} in Section~\ref{sec:lambda-amor-core} and then giving its categorical model in detail in Section~\ref{sec:adjoint-model}.

\paragraph{Interpreting cost annotations and graded modalities}
In \cref{def:lamor:ordered-commutative-monoid} we said that costs must
form an ordered commutative monoid. We thus consider for our model a symmetric
monoidal category $\costcategory$ whose objects are costs, with unit
$0 \in \catobjects{\costcategory}$ and tensor $+ : \costcategory \times \costcategory \to \costcategory$.

Next, we assume some base symmetric
monoidal closed category $\mathbb{X}$ (details to be expanded in
Section~\ref{sec:adjoint-model}) to interpret contexts
via a monoidal product, and the function space via the adjoint
exponent structure. Given that the type system of $\lambda$-amor is based
on the standard idea of an explicit graded monad for cost,
we naturally expect this to be modelled in $\mathbb{X}$ by a
a \emph{strong graded monad} $M : \mathbb{C} \to [\mathbb{X}, \mathbb{X}]$,
i.e.:
$$\interpret{\costmonad{p}{\tau}} \coloneq M(p, \interpret{\tau})$$
Note, we write $M(p,A) : \mathbb{X}$ for the application of
$M$ to two arguments $p \in \mathbb{C}$ and $A \in \mathbb{X}$.
We can thus model the $\mathtt{bind}$ and $\mathtt{ret}$ constructs
using the operations of the graded monad $M$.

\paragraph{Interpreting potential-stored computations}
We can now approach the semantics of the indexed potential modality
$\potentialtype{p}{\tau}$. Let us assume we have some functor $P : \mathbb{C}^\op
\to [\mathbb{X}, \mathbb{X}]$, i.e., a $\mathbb{C}^\op$-indexed family
of $\mathbb{X}$-endofunctors, which we use to interpret the potential types:
$$
\interpret{\potentialtype{p}{\tau}} \coloneq P(p, \interpret{\tau})
$$
Thus, to model $\mathtt{store}$ (typing replayed
on the left), we require a family of morphisms (right):

\begin{minipage}{0.5\textwidth}
  \tyrulestore
\end{minipage}
\begin{minipage}{0.35\textwidth}
  \begin{align*}
  \interpret{\mathtt{store}}_{p, A} : A \to M(p, P(p, A)) \\
\end{align*}
\end{minipage}

\noindent
Such an operation is highly suggestive of the unit operation of a
$\mathbb{C}$-indexed family of adjunctions $P p \dashv M p$, i.e., an
adjunction at every $p \in \mathbb{C}$. If such a
family of adjunctions exists then we would also have the counit operation
$\varepsilon_{p, A} : P (p, M (p, A)) \to A$
which we can see as intuitively corresponding to the idea
of using stored potential $p$ to `pay' for a computational cost of $p$ to
yield a result of type $A$. There is however no primitive in $\lambda$-amor
 that directly corresponds to the adjunction's counit $\varepsilon$.

Let us
hold that thought for the moment, and consider the model of the remaining
primitive that involves interaction between potential and cost,
$\mathtt{release}$:

\begin{minipage}{1\textwidth}
  \tyrulerelease
\end{minipage} \\

\noindent
Let us assume a family of \emph{strong} endofunctors
$P : \mathbb{C}^\op \to [\mathbb{X}, \mathbb{X}]$
such that we have $t_{p, A, B} : A \otimes P(p, B) \to P(p, A \otimes B)$,
and a family of adjunctions $P p \dashv M p$.
By the induction of the interpretation we
assume there exist morphisms $\interpret{\mathcal{D}_1} : \interpret{\Gamma_1} \to P(p_1,\interpret{\tau_1})$ and $\interpret{\mathcal{D}_2} : \interpret{\Gamma_2} \otimes \interpret{\tau_1} \to M(p_1 + p_2, \interpret{\tau_2})$. We can
then construct a derivation as follows using these ingredients but with a
missing morphism, marked below with ???:

          \adjustbox{scale=0.89,center}{\begin{tikzcd}
            \interpret{\Gamma_1} \otimes \interpret{\Gamma_2} \ar[rr, "\sigma_{\interpret{\Gamma_1},\interpret{\Gamma_2}}"] & &
            \interpret{\Gamma_2} \otimes \interpret{\Gamma_1} \ar[r, "\idmorph{\interpret{\Gamma_2}} \otimes \interpret{\mathcal{D}_1}"] &
            \interpret{\Gamma_2} \otimes P(p_1,\interpret{\tau_1}) \ar[dlll, "\mathit{st}_{p_1, \interpret{\Gamma_2},\interpret{\tau_1}}", swap] \\
            P(p_1, \interpret{\Gamma_2} \otimes \interpret{\tau_1}) \ar[rr, "{P(p_1, \interpret{\mathcal{D}_2})}"'] & & P(p_1, M(p_1 + p_2, \interpret{\tau_2})) \ar[r, "???"']
            &
            P(p_1, M(p_1, M(p_2, \interpret{\tau_2}))) \ar[r, "\varepsilon_{p_1}"'] & M(p_2, \interpret{\tau_2})
          \end{tikzcd}} \\

\noindent
The derivation ends with the counit of the adjunction, paying for
cost $p_1$, with the remaining cost of $p_2$ captured by
the graded monad. How do we fill the missing hole? From the signature
we can see that, underneath the $P(p_1, -)$ functor,
we need to \emph{split} up cost $p_1 + p_2$. This
operation matches that of the comultiplication
of a \emph{graded comonad}, i.e.,
$$
\delta_{p_1, p_2, A} : M(p_1 + p_2, A) \to M(p_1, M(p_2, A))
$$
The $\tmrun{e}$ construct, which was added in a later extension of $\lambda$-amor to probabilistic
programs (p$\lambda$-amor~\cite{Rajani2024modal,Rajani2024technical})
then provides the counit operation for the view of $M$ as a graded comonad
since it maps an expression of type $\costmonad{0}{\tau}$ to one of type $\tau$.

We thus arrive at a model that suggests $M$ should be simultaneously a
(strong) graded monad (for combining cost) and graded comonad (for
splitting cost) with a family of adjunctions between each $P$ (as left
adjoint) and $M$, for both storing potential and `paying' for cost. We
return to the details of this model in Section~\ref{sec:adjoint-model}
but in the next section propose our refinement of $\lambda$-amor to
\textbf{$\lambda$-amor Core} based on these category theoretic insights.

\section{A Refinement: Lambda-Amor Core}
\label{sec:lambda-amor-core}
We introduce a more reductive core calculus for cost and potential,
\emph{$\lambda$-amor Core}, which is largely similar to \emph{$\lambda$-amor}
but adapted based on the preceding discussion and also developing in more
detail the equational theory of language in order to consider
the class of sound category theoretic models.

\subsection{Syntax and Typing}
\label{subsec:lamor-core:syntax-typing}

The syntax of $\lambda$-amor Core is that of $\lambda$-amor but
\emph{without} the $\mathtt{release}$ primitive and with the addition
of several new primitives:
\[
  \begin{array}{lclr}
    e & \Coloneqq & \ldots \mid \tmpay{e} \mid \tmplet{x}{e}{e} & \\
    v & \Coloneqq & \ldots \mid \cancel{\tmrelease{x}{e}{e}} \mid \tmsplit{e} &
  \end{array}
\]
The typing of these new primitives is given as follows: \\

  {\scalebox{0.93}{
   \hspace{-1.5em}\begin{minipage}{0.22\textwidth}
    \tyrulepay
  \end{minipage}\hspace{0.8em}
  \begin{minipage}{0.44\textwidth}
    \tyruleplet
  \end{minipage}}
  \begin{minipage}{0.4\textwidth}
    \tyrulesplit
  \end{minipage}} \\

\noindent
The $\tmpay{e}$ construct captures the idea that we can use $p$
stored potential to pay for a computation costing $p$
to yield the resulting $\tau$ value.
The $\mathtt{plet}$ construct provides functoriality of the potential
type.
Lastly, the $\mathtt{split}$ construct allows a costful computation
to split its cost across two layers.

\subsection{Operational Semantics}
\label{subsec:lamor-core:operational-semantics}

We extend the operational semantics of Section~\ref{subsec:lamor:operational-semantics} for the three new constructs. Since $\mathtt{pay}$ and $\mathtt{plet}$
are expression forms, they have pure expression evaluation rules, whereas
$\mathtt{split}$ has only a forcing evaluation rule:

\begin{minipage}{0.3\textwidth}
  \begin{prooftree}
    \AxiomC{$e \evaluatesto v$}
    \AxiomC{$v \forcesto{p} v'$}
    \RightLabel{\textsc{E-Pay}}
    \BinaryInfC{$\tmpay{e} \evaluatesto v'$}
  \end{prooftree}
  \end{minipage}
\begin{minipage}{0.3\textwidth}
  \begin{prooftree}
    \AxiomC{$[e_1/x]e_2 \evaluatesto v$}
    \RightLabel{\textsc{E-PLet}}
    \UnaryInfC{$\tmplet{x}{e_1}{e_2} \evaluatesto v$}
  \end{prooftree}
  \end{minipage}
\begin{minipage}{0.3\textwidth}

\begin{prooftree}
  \AxiomC{$e \evaluatesto v$}
  \AxiomC{$v \forcesto{p} v'$}
  \RightLabel{\textsc{F-Split}}
  \BinaryInfC{$\tmsplit{e} \forcesto{p} v'$}
\end{prooftree}
\end{minipage}


\subsection{The Decomposition of Release}
\label{subsec:lamor:differences}

Our refactoring of $\lambda$-amor into $\lambda$-amor Core
can be seen simply as the decomposition:
\begin{equation*}
  \mathbf{Release} \equiv \mathbf{Pay} + \mathbf{Split} + \mathbf{PLet}.
\end{equation*}
That is, we replace the complicated expression $\tmrelease{x}{e_1}{e_2}$ that is present in the original system by three more primitive
terms $\tmpay{e}$, $\tmplet{x}{e_1}{e_2}$ and $\tmsplit{e}$.
In the following, we describe how to express each of the three constructs on the right in terms of release, and how to express release in terms of the three new constructs.

The relationship between the first three is made more precise in \cref{sec:kripke-model}, Theorem~\ref{thrm:kripke-model-instance}, where we use
the syntactic Kripke model of $\lambda$-amor given by \cite{Rajani2021lamor} to give a model of $\lambda$-amor Core, showing
that $\tmpay{e}$, $\tmplet{x}{e_1}{e_2}$ and $\tmsplit{e}$ are exactly modelled by the derivations in terms of
$\mathtt{release}$ shown below.

\subsubsection{Expressing Pay in terms of Release}
\label{subsubsec:pay-via-release}
\phantom{x}

\begin{minipage}{0.25\textwidth}
  \tyrulepay
\end{minipage}
\begin{minipage}{0.7\textwidth}
  \begin{prooftree}
    \AxiomC{$\Gamma \vdash e : \potentialtype{p}{\costmonad{p}{\tau}}$}
    \AxiomC{}
    \RightLabel{\textsc{Var}}
    \UnaryInfC{$x : \costmonad{(p + 0)}{\tau} \vdash x : \costmonad{(p + 0)}{\tau}$}
    \RightLabel{\textsc{Release}}
    \BinaryInfC{$\Gamma \vdash \tmrelease{x}{e}{x} : \costmonad{0}{\tau}$}
    \RightLabel{\textsc{Run}}
    \UnaryInfC{$\Gamma \vdash \tmrun{\tmrelease{x}{e}{x}} : \tau$}
  \end{prooftree}
\end{minipage}
\vspace{0.2cm}

\subsubsection{Expressing Split in terms of Release}
\label{subsubsec:split-via-release}

\tyrulesplit

\noindent
Let $\Gamma_x = x : \potentialtype{p_1}{(\costmonad{(p_1 + p_2)}{\tau})}$ and
$\Gamma_y = y : \costmonad{(p_1 + p_2)}{\tau}$, then, we derive:

{\scalebox{0.9}{
\begin{minipage}{1\linewidth}
\begin{prooftree}
  \AxiomC{$\Gamma \vdash e : \costmonad{(p_1 + p_2)}{\tau}$}
  \UnaryInfC{$\Gamma \vdash \tmstore{e} : \costmonad{p_1}{(\potentialtype{p_1}{(\costmonad{(p_1 + p_2)}{\tau})})}$}
  \AxiomC{}
  \UnaryInfC{$\Gamma_x \vdash x : \potentialtype{p_1}{(\costmonad{(p_1 + p_2)}{\tau})}$}
  \AxiomC{}
  \UnaryInfC{$\Gamma_y \vdash y : \costmonad{(p_1 + p_2)}{\tau}$}
  \BinaryInfC{$\Gamma_x \vdash \tmrelease{y}{x}{y} : \costmonad{p_2}{\tau}$}
  \UnaryInfC{$\Gamma_x \vdash \tmret{\tmrelease{y}{x}{y}} : \costmonad{0}{\costmonad{p_2}{\tau}}$}
  \BinaryInfC{$\Gamma \vdash \tmbind{x}{\tmstore{e}}{\tmret{(\tmrelease{y}{x}{y})}} : \costmonad{p_1}{(\costmonad{p_2}{\tau})}$}
\end{prooftree}
\end{minipage}
}}

\subsubsection{Expressing PLet in terms of Release}

\tyruleplet

\begin{prooftree}
  \AxiomC{$\Gamma_1 \vdash e_1 : \potentialtype{p}{\tau_1}$}
  \AxiomC{$\Gamma_2, x : \tau_1 \vdash e_2 : \tau_2$}
  \RightLabel{\textsc{Store}}
  \UnaryInfC{$\Gamma_2, x : \tau_1 \vdash \tmstore{e_2} : \costmonad{(p + 0)}{(\potentialtype{p}{\tau_2})}$}
  \RightLabel{\textsc{Release}}
  \BinaryInfC{$\Gamma_1,\Gamma_2 \vdash \tmrelease{x}{e_1}{\tmstore{e_2}} : \costmonad{0}{(\potentialtype{p}{\tau_2})}$}
  \RightLabel{\textsc{Run}}
  \UnaryInfC{$\Gamma_1,\Gamma_2 \vdash \tmrun{\tmrelease{x}{e_1}{\tmstore{e_2}}} : \potentialtype{p}{\tau_2}$}
\end{prooftree}

\subsubsection{Expressing Release in terms of Split, PLet and Pay}

Lastly, to see that we can recover the original release construct we have the following dervation:

\tyrulerelease

\begin{prooftree}
  \AxiomC{$\Gamma_1 \vdash e_1 : \potentialtype{p_1}{\tau_1}$}
  \AxiomC{$\Gamma_2, x : \tau_1 \vdash e_2 : \costmonad{(p_1 + p_2)}{\tau_2}$}
  \RightLabel{\textsc{Split}}
  \UnaryInfC{$\Gamma_2, x : \tau_1 \vdash \tmsplit{e_2} : \costmonad{p_1}{(\costmonad{p_2}{\tau_2})}$}
  \RightLabel{\textsc{PLet}}
  \BinaryInfC{$\Gamma_1,\Gamma_2 \vdash \tmplet{x}{e_1}{e_2} : \potentialtype{p_1}{(\costmonad{p_1}{(\costmonad{p_2}{\tau_2})})}$}
  \RightLabel{\textsc{Pay}}
  \UnaryInfC{$\Gamma_1,\Gamma_2 \vdash \tmpay{\tmplet{x}{e_1}{\tmsplit{e_2}}} : \costmonad{p_2}{\tau}$}
\end{prooftree}

\subsection{Equational Theory}
\label{subsec:lamor:equational}

In this section we introduce a typed equational theory for $\lambda$-amor Core which is based on the type system described in \cref{sec:lamor}, combined with the refinement described in \cref{subsec:lamor-core:syntax-typing}.
The core judgement form of this equational theory is $\mathcal{D}_1 \equiv \mathcal{D}_2$, where $\mathcal{D}_1 : (\Gamma \vdash e_1 : \tau)$ and $\mathcal{D}_2 : (\Gamma \vdash e_2 : \tau)$. We assume that there is no variable capture involved in any such
equations, i.e., that all variable binders do not clash with free names elsewhere.
That is, the two typing derivations must agree in the context and type, but might differ in the expression being typed.
We present a selection of the rules in this section and focus on those rules which characterize cost, potential and their interaction.
The parts that we omit, i.e. congruences, uses of subsumption, and $\beta$- and $\eta$-rules for ordinary types, are completely standard.
The complete set of rules is provided in \cref{sec:lambda-amor-core-complete}.

In order to formulate some of the laws we need to use the following derived operation which corresponds to the functorial map for cost.
\begin{prooftree}
  \AxiomC{$\Gamma_1 \vdash e_1 : \costmonad{p}{\tau_1}$}
  \AxiomC{$\Gamma_2, x : \tau_1 \vdash e_2 : \tau_2$}
  \RightLabel{\textsc{CLet}}
  \BinaryInfC{$\Gamma_1,\Gamma_2 \vdash \tmclet{x}{e_1}{e_2} : \costmonad{p}{\tau_2}$}
\end{prooftree}
We can define $\tmclet{x}{e_1}{e_2} \coloneq \tmbind{x}{e_1}{\tmret{e_2}}$; the name $\mathtt{clet}$ is chosen to correspond to $\mathtt{plet}$, the functorial map for potential.
Some of the laws are given twice using both $\mathtt{clet}$ and its expansion in terms of the bind and return operations.

\subsubsection{Laws Characterizing Cost as a Monad}

The first set of laws characterize $\costmonad{p}{\tau}$ as a graded monad.
We have two unit laws, and one law which witnesses the associativity of the bind operation.
\begin{align}
    \tag*{$\equiv$\textsc{-Monad-Unit-L}}
    \textsc{Bind}(\textsc{Ret}(\mathcal{D}_1), \mathcal{D}_2) &\equiv \textsc{Subst}_1(\mathcal{D}_2, \mathcal{D}_1) \\
    \tag*{$\equiv$\textsc{-Monad-Unit-R}}
    \textsc{Bind}(\mathcal{D},\textsc{Ret}(\textsc{Var})) &\equiv \mathcal{D} \\
    \tag*{$\equiv$\textsc{-Monad-Assoc}}
    \textsc{Bind}(\textsc{Bind}(\mathcal{D}_1, \mathcal{D}_2), \mathcal{D}_3) &\equiv \textsc{Bind}(\mathcal{D}_1, \textsc{Ex}^\ast(\textsc{Bind}(\mathcal{D}_2, \mathcal{D}_3)))
\end{align}
Note that we have $\textsc{Bind}(\textsc{Ret}(\textsc{Var}), \mathcal{D}) \equiv \mathcal{D}$ as an instance of $\equiv$\textsc{-Monad-Unit-L}.
The following derivations show why we need to use exchange in the equation for $\equiv$\textsc{-Monad-Assoc}.

\begin{prooftree}
    \AxiomC{$\Gamma_1 \vdash e_1 : \costmonad{p_1}{\tau_1}$}
    \AxiomC{$\Gamma_2, x : \tau_1 \vdash \costmonad{p_2}{\tau_2}$}
    \BinaryInfC{$\Gamma_1, \Gamma_2 \vdash \tmbind{x}{e_1}{e_2} : \costmonad{(p_1 + p_2)}{\tau_2}$}
    \AxiomC{$\Gamma_3, y : \tau_2 \vdash e_3 : \costmonad{p_3}{\tau_3}$}
    \BinaryInfC{$\Gamma_1, \Gamma_2, \Gamma_3 \vdash \tmbind{y}{(\tmbind{x}{e_1}{e_2})}{e_3} : \costmonad{((p_1 + p_2) + p_3)}{\tau_3}$}
\end{prooftree}
\begin{center}
    $\equiv$
\end{center}
\begin{prooftree}
    \AxiomC{$\Gamma_1 \vdash e_1 : \costmonad{p_1}{\tau_1}$}
    \AxiomC{$\Gamma_2, x : \tau_1 \vdash \costmonad{p_2}{\tau_2}$}
    \AxiomC{$\Gamma_3, y : \tau_2 \vdash e_3 : \costmonad{p_3}{\tau_3}$}
    \BinaryInfC{$\Gamma_2, x : \tau_1, \Gamma_3 \vdash \tmbind{y}{e_2}{e_3} : \costmonad{(p_2 + p_3)}{\tau_3}$}
    \doubleLine
    \UnaryInfC{$\Gamma_2, \Gamma_3, x : \tau_1 \vdash \tmbind{y}{e_2}{e_3} : \costmonad{(p_2 + p_3)}{\tau_3}$}
    \BinaryInfC{$\Gamma_1, \Gamma_2, \Gamma_3 \vdash \tmbind{x}{e_1}{(\tmbind{y}{e_2}{e_3})} : \costmonad{(p_1 + (p_2 + p_3))}{\tau_3}$}
\end{prooftree}
Note that our no-variable-capture condition means here that
$x$ is not free in $e_3$, avoiding variable capture on the right-hand side
of this equality.

We characterise the behaviour of tick and its interaction with the graded monad as:
\begin{align}
    \tag*{$\equiv$\textsc{tick}$_1$}
    \textsc{Tick}^0 & \equiv \textsc{Ret}(\textsc{Unit}) \\
    \tag*{$\equiv$\textsc{tick}$_2$}
\textsc{Bind}(\textsc{Tick}^p, \textsc{Tick}^q) & \equiv \textsc{Tick}^{p+q}
\end{align}

\subsubsection{Laws Characterizing Cost as a Comonad}

Since $\costmonad{p}{\tau}$ is also a graded comonad we require the following additional laws, where we have for
$\equiv$-\textsc{Comonad}$_1$ that $\mathcal{D}: (\Gamma \vdash e : \costmonad{(0 + p)}{\tau})$, for $\equiv$-\textsc{Comonad}$_2$ that $\mathcal{D} : (\Gamma \vdash e : \costmonad{(p + 0)}{\tau})$, and for
$\equiv$-\textsc{Comonad}$_3$ we have that $\mathcal{D} : (\Gamma \vdash e : \costmonad{(p_1 + p_2 + p_3)}{\tau})$.
\begin{align*}
  \tag*{$\equiv$-\textsc{Comonad}$_1$}
  \textsc{Run}(\textsc{Split}(\mathcal{D})) &\equiv \mathcal{D} \\
  \tag*{$\equiv$-\textsc{Comonad}$_2$}
  \textsc{CLet}(\textsc{Split}(\mathcal{D}), \textsc{Run}(\textsc{Var})) &\equiv \mathcal{D} \\
  \tag*{$\equiv$-\textsc{Comonad}$_3$}
  \textsc{Split}(\textsc{Split}(\mathcal{D})) &\equiv \textsc{CLet}(\textsc{Split}(\mathcal{D}), \textsc{Split}(\textsc{Var}))
\end{align*}

\subsubsection{Laws Characterizing Potential as a Functor}
The following two equations ensure that the $\mathtt{plet}$ construct behaves as the functorial action for the potential type.
The second equation requires that the variable bound by the inner $\mathtt{plet}$ on the left side does not occur free in the derivation $\mathcal{D}_3$.
\begin{align*}
  \tag*{$\equiv$-\textsc{PLet}$_1$}
  \textsc{PLet}(\mathcal{D}, \textsc{Var})&\equiv \mathcal{D} \\
  \tag*{$\equiv$-\textsc{PLet}$_2$}
  \textsc{PLet}(\textsc{PLet}(\mathcal{D}_1,\mathcal{D}_2),\mathcal{D}_3) &\equiv \textsc{PLet}(\mathcal{D}_1,\textsc{PLet}(\mathcal{D}_2, \mathcal{D}_3))
\end{align*}
\subsubsection{The Relationship between the Monadic and Comonadic Structure}

Lastly, the laws $\equiv$\textsc{-Run-Ret}, $\equiv$\textsc{-Ret-Run}, $\equiv$\textsc{-Split-Bind} and $\equiv$\textsc{-Bind-Split} witness the isomorphism of the unit/counit and multiplication/comultiplication.
\begin{align}
    \tag*{$\equiv$\textsc{-Run-Ret}}
    \textsc{Run}(\textsc{Ret}(\mathcal{D})) \equiv \mathcal{D} \\
    \tag*{$\equiv$\textsc{-Ret-Run}}
    \textsc{Ret}(\textsc{Run}(\mathcal{D})) \equiv \mathcal{D} \\
    \tag*{$\equiv$\textsc{-Split-Bind}}
    \textsc{Split}(\textsc{Bind}(\mathcal{D},\textsc{Var})) &\equiv \mathcal{D} \\
    \tag*{$\equiv$\textsc{-Bind-Split}}
    \textsc{Bind}(\textsc{Split}(\mathcal{D}),\textsc{Var}) &\equiv \mathcal{D}
\end{align}

\subsubsection{Laws Characterizing the Adjunction}

Lastly, we have equations about the interaction between $\mathtt{pay}$ and
$\mathtt{store}$, using also the functorial action of $\mathtt{clet}$
and $\mathtt{plet}$:

\begin{tikzcd}
    \costmonad{p}{\tau} \ar[rdd, "\mathrm{id}", swap] \ar[r, "\mathtt{store}"] &
    \costmonad{p}{(\potentialtype{p}{(\costmonad{p}{\tau})})} \ar[dd, "\mathtt{clet}(\mathtt{pay})"] &
    \potentialtype{p}{\tau} \ar[rrrrdd, "\mathrm{id}", swap] \ar[rrrr, "\mathtt{plet}(\mathtt{store})"] &&&&
    \potentialtype{p}{(\costmonad{p}{(\potentialtype{p}{\tau})})} \ar[dd, "\mathtt{pay}"] \\
    \\
    & \costmonad{p}{\tau} &&&&& \potentialtype{p}{\tau}
\end{tikzcd}

\begin{align*}
  \tag*{$\equiv$\textsc{-Adjunction}$_1$}
  \textsc{CLet}(\textsc{Store}(\mathcal{D}),\textsc{Pay}(\textsc{Var})) &\equiv \mathcal{D} \\
  \tag*{$\equiv$\textsc{-Adjunction}$_2$}
  \textsc{Pay}(\textsc{PLet}(\mathcal{D}, \textsc{Store}(\textsc{Var}))) &\equiv \mathcal{D}
\end{align*}

\section{Adjoint Model of Lambda-Amor Core}
\label{sec:adjoint-model}
In this section, we now expand the detail of the model discussion in Section~\ref{sec:lamor-cat-model}, giving a categorical model for the $\lambda$-amor Core calculus introduced in~\cref{sec:lambda-amor-core}). We begin by replaying the (standard) definitions of \emph{graded monads}, \emph{graded comonads},
and \emph{strong functors}.

\begin{definition}[Graded Monad]~\cite{Katsumata2014parametric,Orchard2014marriage}
    \label{def:adjoint-model:graded-monad}
    Given a strict monoidal category $\costcategory$,
    with bifunctor $+ : \mathbb{C} \times \mathbb{C} \to \mathbb{C}$
    and unit object $0 \in \catobjects{\costcategory}$, then a
    $\costcategory$-graded monad $M$ over category $\modelcategory$ consists of:
    \begin{enumerate}
        \item A bifunctor $M : \costcategory \times \modelcategory \to \modelcategory$
        (which in some presentations is given isomorphically in `curried' form as
        $\costcategory \to [\modelcategory, \modelcategory]$, i.e., as an indexed
        family of endofunctors on $\modelcategory$);

        \item A \emph{unit} natural transformation $\eta_X : X \to M(0,X)$;

        \item A \emph{multiplication} a natural transformation $\mu_{X,p_1,p_2} : M(p_1, M(p_2,X)) \to M(p_1 + p_2, X)$
    such that the following diagrams commute:

    \hspace{-3em}\begin{minipage}{0.4\linewidth}\begin{tikzcd}
       M(p_1,X) \ar[d, "{M(p_1, X)}" swap] \ar[r, "\eta_X"] & M(0, M(p_1, X)) \ar[d, "{\mu_{0,p_1,X}}"] \\
       M(p_1, M(0, X)) \ar[r, "{\mu_{p_1,0,X}}" swap] & M(p_1, X)
    \end{tikzcd}
    \end{minipage}
    \hspace{1em}
    \begin{minipage}{0.5\linewidth}
      \begin{tikzcd}
        M(p_1,M(p_2,M(p_3,X))) \ar[d, "{M(p_1, \mu_{p_2,p_3,X})}" swap] \ar[rr, "{\mu_{p_1,p_2,M(p_3, X)}}"] & & M(p_1 + p_2, M(p_3,X)) \ar[d, "{\mu_{p_1 + p_2, p_3, X}}"] \\
        M(p_1, M(p_2 + p_3, X)) \ar[rr, "{\mu_{p_1,p_2+p_3,X}}" swap] & &  M(p_1 + p_2 + p_3, X)
      \end{tikzcd}
      \end{minipage}
    \end{enumerate}
\end{definition}

The dual of a graded monad is a graded comonad.
\begin{definition}[Graded Comonad]~\cite{DBLP:conf/icalp/PetricekOM13,DBLP:conf/icfp/PetricekOM14}
  \label{def:adjoint-model:graded-comonad}
  Given a strict monoidal category $(\costcategory, +, 0)$, then a
   $\costcategory$-graded comonad $M$ in $\modelcategory$ consists of:
    \begin{enumerate}
        \item A bifunctor $M : \costcategory^\op \times \modelcategory \to \modelcategory$
        \item A \emph{counit} natural transformation $\varepsilon_X : M(0,X) \to X$
        \item A \emph{comultiplication} natural transformation $\delta_{p_1,p_2,X} : M(p_1 + p_2, X) \to M(p_1, M(p_2,X))$
    \end{enumerate}
    The diagrams required to commute are exactly the dual diagrams of \cref{def:adjoint-model:graded-monad}.
\end{definition}
Furthermore, we will require the additional property of \emph{tensorial strength}
for graded (co)monads:
\begin{definition}[Strong Functor]
    \label{def:adjoint-model:strong-functor}
    Given a monoidal category $\modelcategory$,
   a \emph{strong functor} is an endofunctor $F : \modelcategory \to \modelcategory$ with a (left) strength natural transformation:
    \begin{equation*}
        {\mathit{st}}_{X,Y} : X \otimes F(Y) \to F(X \otimes Y)
    \end{equation*}
    such that the following two diagrams commute, i.e., that tensorial strength commutes
    with the associator $\alpha$ and left unitor $\lambda$ of the monoidal category:

    \begin{minipage}{0.66\linewidth}
      \begin{tikzcd}
        (X \otimes Y) \otimes F(Z) \ar[d, "\alpha_{X,Y,F(Z)}", swap]\ar[rr, "\mathit{st}_{X \otimes Y, Z}"]& & F((X \otimes Y) \otimes Z) \ar[d, "F(\alpha_{X,Y,Z})"]\\
        X \otimes (Y \otimes F(Z)) \ar[r, "\idmorph{X}\otimes \mathit{st}_{Y,Z}"] & X \otimes F(Y \otimes Z) \ar[r,"\mathit{st}_{X,Y \otimes Z}"] & F(X \otimes (Y \otimes Z))
      \end{tikzcd}
    \end{minipage}
    \begin{minipage}{0.3\linewidth}
    \begin{tikzcd}
        1 \otimes F(X)\ar[r, "\mathit{st}_{1,X}"] \ar[rd, "\lambda_X" swap]& F(1 \otimes X) \ar[d, "F(\lambda_X)"] \\
        & F(X)
     \end{tikzcd}
    \end{minipage}
\end{definition}

\begin{definition}[Strong Graded Monad]~\cite{Katsumata2014parametric}
    \label{def:adjoint-model:strong-graded-monad}
    A strong graded monad is a graded monad where each endofunctor $M(p, -)$ is also a
    strong functor, i.e., there is $\costcategory$-family of natural transformations:
    \begin{equation*}
        \mathit{st}_{p,X,Y} : X \otimes M(p,Y) \to M(p,X \otimes Y)
    \end{equation*}
    which commutes with the operations of the graded monad (omitted for brevity; see \cref{apx:model-details:strong-graded-monad-diagrams}).
\end{definition}

\subsection{Description of the Model}

We now turn to the model of $\lambda$-amor Core, which has two defining categories,
the \emph{cost category} and \emph{model category}:
\begin{definition}[Cost category $\costcategory$]
  \label{def:adjoint-model:costcategory}
  The cost category $\costcategory$ is a \emph{thin, strict symmetric monoidal category}:
  \begin{enumerate}
    \item Thinness means that there is at most one morphism between any two objects of $\costcategory$, which models the pre-order relation between costs.
    \item Monoidality means that there is an object $0 \in \catobjects{\costcategory}$ called the unit
          and a bifunctor $+ : \mathbb{C} \times \mathbb{C} \to \mathbb{C}$ which interpret the unit cost and the addition of two costs, respectively. Strictness means that
          associativity and unitality are up-to equality.
    \item Symmetric means that $p_1 + p_2 = p_2 + p_1$.
  \end{enumerate}
\end{definition}

\begin{definition}[Model category $\modelcategory$]
  \label{def:adjoint-model:model-category}
    The category $\modelcategory$ is equipped with the following structure:
    \begin{enumerate}
      \item It has a symmetric monoidal structure given by the categorical product and terminal element:
        \begin{equation*}
          (\modelcategory, \times : \modelcategory \times \modelcategory \to \modelcategory, 1 \in \catobjects{\modelcategory}, \alpha, \lambda, \rho, \sigma)
        \end{equation*}
        Here, and in the following, $\alpha, \lambda, \rho$ and $\sigma$ are the natural transformations which witnesses
        associativity, the left and right unit laws, and the symmetry, respectively.
        This structure is used for interpreting the negative product type $\tau_1 \& \tau_2$ and its unit $\tyunitn$.
      \item It has a symmetric monoidal structure given by the categorical sum and initial element:
        \begin{equation*}
          (\modelcategory, + : \modelcategory \times \modelcategory \to \modelcategory, 0 \in \catobjects{\modelcategory}, \alpha, \lambda, \rho, \sigma)
        \end{equation*}
        This structure is used for interpreting the positive disjunction $\tau_1 \oplus \tau_2$ and its unit $\tyvoid$.
      \item It has a third symmetric monoidal structure
        \begin{equation*}
          (\modelcategory, \otimes : \modelcategory \times \modelcategory \to \modelcategory, 1 \in \catobjects{\modelcategory}, \alpha, \lambda, \rho, \sigma)
        \end{equation*}
        that is used for interpreting the positive product type $\tau_1 \otimes \tau_2$ and its unit $\tyunitp$, as well as for modelling typing contexts $\Gamma$.
      \item The monoidal structure given by $\otimes$ has a right adjoint $\to$:
        \begin{equation*}
           \phi_{\to} : \homset{\modelcategory}{Z \otimes Y}{X} \simeq \homset{\modelcategory}{Z}{Y \to X} : \phi^{-1}_{\to}
        \end{equation*}
        that is used for interpreting the function type $\tau_1 \multimap \tau_2$.

\item We have a strong $\costcategory$-graded monad $M : \costcategory \times \modelcategory \to \modelcategory$ used for interpreting the type $\costmonad{p}{\tau}$
      (see Definition~\ref{def:adjoint-model:graded-monad}).

\item There is a family of maps $\uparrow_p : 1 \to M(p, 1)$ satisying the following equations, meaning that $\uparrow_p$ is
a \emph{graded algebraic operation}:
\begin{equation*}
\eta_1 = \uparrow_0 \qquad  \mu_{p,q,1} \circ M(p, \uparrow_{q}) \circ \uparrow_{p} = \uparrow_{p+q}
\end{equation*}

\item The bifunctor $M$ also has the structure of $\costcategory^\op$-graded comonad
      (see Definition~\ref{def:adjoint-model:graded-comonad}). Note the double $\op$ means
      that $M$ is always covariant (in both arguments).

\item The graded monadic and graded comonadic structure are \emph{compatible} meaning
that multiplication and comuliplication form an isomorphism, and unit and counit form an isomorphism:
        \begin{align*}
          \eta_X : X &\simeq M(0,X) : \eta^{-1}_X = \varepsilon_{X} \\
          \mu_{p_1,p_2,X} : M(p_1, M(p_2,X)) &\simeq M(p_1 + p_2, X) : \mu^{-1}_{p_1,p_2,X} = \delta_{p_1,p_2,X}
        \end{align*}

      \item We have a strong functor $P : \costcategory^\op \times \modelcategory \to \modelcategory$ (see  Defintion~\ref{def:adjoint-model:strong-functor}), for interpreting the potential type, that is contravariant in its first argument.

      \item For every grade $p \in \catobjects{\costcategory}$, there is an adjunction  $P(p,-) \dashv M(p,-)$, thus:
        \begin{equation*}
          \phi_{p,X,Y} : \homset{\modelcategory}{P(p,X)}{Y} \simeq \homset{\modelcategory}{X}{M(p,Y)} : \phi^{-1}_{p,X,Y}
        \end{equation*}
        Alternatively, but equivalently, we can state the family
        of adjunctions via a pair of natural transformations $\eta_{p,X} : X \to M(p, P(p, X))$ and $\epsilon_{p,X} : P(p, M(p, X)) \to X$ along with the usual `triangle identities'.
    \end{enumerate}
\end{definition}


\subsection{Interpretation of Types, Contexts and Subtyping}

\begin{definition}[Interpretation of Types and Typing Contexts]
    The interpretation of types and typing contexts in the model is given by:
    \begin{gather*}
      \begin{aligned}
        \interpret{\tau_1 \otimes \tau_2} &\coloneq \interpret{\tau_1} \otimes \interpret{\tau_2} & \interpret{\tyunitp} &\coloneq 1 & \interpret{\tau_1 \multimap \tau_2} &\coloneq \interpret{\tau_1} \to \interpret{\tau_2} \\
        \interpret{\tau_1 \with \tau_2} &\coloneq \interpret{\tau_1} \times \interpret{\tau_2} & \interpret{\tyunitn} &\coloneq 1 & \interpret{\costmonad{p}{\tau}} &\coloneq M(p,\interpret{\tau}) \\
        \interpret{\tau_1 \oplus \tau_2} &\coloneq \interpret{\tau_1} + \interpret{\tau_2} & \interpret{\tyvoid} &\coloneq 0 & \interpret{\potentialtype{p}{\tau}} &\coloneq P(p,\interpret{\tau}) \\
      \end{aligned}\\
      \begin{aligned}
        \interpret{\cdot} &\coloneq 1 &
        \interpret{\Gamma, x : \tau} &\coloneq \interpret{\Gamma} \otimes \interpret{\tau}
      \end{aligned}
    \end{gather*}
\end{definition}

\begin{remark}[Units for positive and negative products]
    In this model the unit for positive and negative products are both interpreted as the terminal element $1$.
    This is possible due to the affineness of the theory, in the linear setting without weakening they have to be distinguished.
\end{remark}

The first thing that we have to verify is that the interpretation of types is compatible with the syntactic subtyping relationship that we have defined.
The following lemma witnesses this compatibility.

\begin{restatable}[Subtyping in the Model]{lemma}{subtypinginmodel}
    If there is a subtyping derivation $\mathcal{S}  : \tau_1 \sub \tau_2$, then there exists a corresponding morphism $\interpret{\mathcal{S}}$ from $\interpret{\tau_1}$ to $\interpret{\tau_2}$.
    Similarly, if there is a subtyping derivation $\mathcal{X} : \Gamma_1 \sub \Gamma_2$, then there is a morphism $\interpret{\mathcal{X}}$ in the model from $\interpret{\Gamma_1}$ to $\interpret{\Gamma_2}$.
    \label{lem:adjoint-model:subtyping-contexts}
    \label{lem:adjoint-model:subtyping}
\end{restatable}
\begin{proof}
    See \cref{apx:model-details:subtyping}.
\end{proof}

\subsection{Interpretation of Typing Derivations}

We now give an interpretation of the typing derivations presented in \cref{subsec:lamor:typing-rules}.
Every such typing derivation is interpreted as a morphism from the interpretation of the context to the interpretation of the type of the judgement.

\begin{theorem}[Fundamental Theorem]
    If $\mathcal{D} : \Gamma \vdash e : \tau$ is derivable, then $\interpret{\Gamma} \overset{\interpret{\mathcal{D}}}{\rightarrow} \interpret{\tau}$.
    \label{thm:adjoint-model:fundamental}
\end{theorem}
\begin{proof}
    By induction on the typing derivation; we begin with the three core rules which characterize the structural properties of the system:
    \begin{description}
        \item[Var] We have to find a morphism from $\interpret{\tau}$ to $\interpret{\tau}$, which is given by the identity morphism $\idmorph{\interpret{\tau}}$.
        \item[Wk] By the induction hypothesis we have a morphism $\interpret{\mathcal{D}} : \interpret{\Gamma} \to \interpret{\tau}$.
        We construct the required morphism as follows, where $\rho$ is the right-unit law of monoidal categories:
        \begin{center}
          \begin{tikzcd}
            \interpret{\Gamma} \otimes \interpret{\tau} \ar[r, "\idmorph{\interpret{\Gamma}} \otimes !"] &
            \interpret{\Gamma} \otimes 1 \ar[r, "\rho_{\interpret{\Gamma}}"] &
            \interpret{\Gamma} \ar[r, "\interpret{\mathcal{D}}"] &
            \interpret{\tau}
          \end{tikzcd}
        \end{center}
        This is the only place in the proof where we use the property that the unit of the monoidal structure of $\otimes$ is terminal, and the only rule which makes the system affine instead of linear.
        \item[Sub] By the induction hypothesis we have a morphism $\interpret{\mathcal{D}} : \interpret{\Gamma_2} \to \interpret{\sigma}$.
        By \cref{lem:adjoint-model:subtyping,lem:adjoint-model:subtyping-contexts} we have morphisms $\interpret{\mathcal{S}} : \interpret{\sigma} \to \interpret{\tau}$ and $\interpret{\mathcal{X}} : \interpret{\Gamma_1} \to \interpret{\Gamma_2}$.
        We can thus define:
        \begin{center}
          \begin{tikzcd}
            \interpret{\Gamma_1} \ar[r, "\interpret{\mathcal{X}}"] &
            \interpret{\Gamma_2} \ar[r, "\interpret{\mathcal{D}}"] &
            \interpret{\sigma} \ar[r, "\interpret{\mathcal{S}}"] &
            \interpret{\tau}
          \end{tikzcd}
        \end{center}
        \item[Ex] Since we require a symmetric monoidal category we can use the braiding $B_{\interpret{\tau_1},\interpret{\tau_2}} : \interpret{\tau_1} \otimes \interpret{\tau_2} \simeq \interpret{\tau_2} \otimes \interpret{\tau_1}$ to perform the exchange of the two types in the context.
    \end{description}
    Next, we show the simple cases for the additive connectives $\with$ and $\oplus$ and their units $\top$ and $0$.
    \begin{description}
        \item[With] From the induction hypothesis there exist morphisms $\interpret{\mathcal{D}_1} : \interpret{\Gamma} \to \interpret{\tau_1}$ and $\interpret{\mathcal{D}_2} : \interpret{\Gamma} \to \interpret{\tau_2}$.
        By the universal property of the categorical product, we can combine those to obtain a morphism $\interpret{\mathcal{D}_1} \triangle \interpret{\mathcal{D}_2} : \interpret{\Gamma} \to \interpret{\tau_1} \times \interpret{\tau_2}$, whose codomain agrees with the
        interpretation of $\interpret{\tau_1 \with \tau_2}$.
        \item[Fst  and Snd] We show the case for the rule \textsc{Fst}; the case for \textsc{Snd} is similar.
        By the induction hypothesis, there exists a morphism $\interpret{\mathcal{D}} : \interpret{\Gamma} \to \interpret{\tau_1 \with \tau_2}$.
        Since the interpretation $\interpret{\tau_1 \with \tau_2}$ is defined to be $\interpret{\tau_1} \times \interpret{\tau_2}$, we can compose with the first projection of categorical products to obtain $\pi_1 \circ \interpret{\mathcal{D}} : \interpret{\Gamma} \to \interpret{\tau_1}$.
        \item[Inl and Inr] We show the case for the rule \textsc{Inl}; the case for \textsc{Inr} is similar.
        By the induction hypothesis, there exists a morphism $\interpret{\mathcal{D}} : \interpret{\Gamma} \to \interpret{\tau}$.
        Since the interpretation $\interpret{\tau_1 \oplus \tau_2}$ is defined to be $\interpret{\tau_1} + \interpret{\tau_2}$, we can compose with the first injection into the categorical coproduct to obtain $\iota_1 \circ \interpret{\mathcal{D}} : \interpret{\Gamma} \to \interpret{\tau_1} + \interpret{\tau_2}$, whose codomain is the same as $\interpret{\tau_1 \oplus \tau_2}$.
        \item[Case-Sum] By the induction hypothesis there exist morphisms $\interpret{\mathcal{D}_1} : \interpret{\Gamma_1} \to \interpret{\tau_1 \oplus \tau_2}$, $\interpret{\mathcal{D}_2} : \interpret{\Gamma_2} \otimes \interpret{\tau_1} \to \interpret{\tau_3}$ and $\interpret{\mathcal{D}_3} : \interpret{\Gamma_2} \otimes \interpret{\tau_2} \to \interpret{\tau_3}$.
        Because $\otimes$ is left adjoint to $\to$ it preserves colimits, and $\interpret{\tau_1 \oplus \tau_2} \otimes \interpret{\Gamma_2}$ is hence isomorphic to $(\interpret{\tau_1} \otimes \interpret{\Gamma_2}) + (\interpret{\tau_2} \otimes \interpret{\Gamma_2})$.
        \begin{center}
          \begin{tikzcd}
            \interpret{\Gamma_1} \otimes \interpret{\Gamma_2} \ar[rr, "\interpret{\mathcal{D}_1} \otimes \idmorph{\interpret{\Gamma_2}}"] & &
            \interpret{\tau_1 \oplus \tau_2} \otimes \interpret{\Gamma_2} \ar[r, dash, "\simeq"] &
            (\interpret{\tau_1} \otimes \interpret{\Gamma_2}) + (\interpret{\tau_2} \otimes \interpret{\Gamma_2}) \ar[rr, "\interpret{\mathcal{D}_2} \triangledown \interpret{\mathcal{D}_3}"] & &
            \interpret{\tau_3}
          \end{tikzcd}
        \end{center}
        \item[Unit-N] We have to construct a morphism from $\interpret{\Gamma}$ to $\interpret{\top}$, which is defined to be the terminal object in the category. We can therefore pick the unique morphism to the terminal object.
        \item[Void] By the induction hypothesis there exists a morphism $\interpret{\mathcal{D}} : \interpret{\Gamma} \to \interpret{0}$.
        Since $\interpret{0}$ is the initial object in the category, we can define $! \circ \interpret{\mathcal{D}} : \interpret{\Gamma} \to \interpret{\tau}$.
    \end{description}
    Next, let us look at the multiplicative connectives $\otimes$, $\multimap$ and the unit $\tyunitp$.
    \begin{description}
        \item[Tensor] By the induction hypothesis, there exist morphisms $\interpret{\mathcal{D}_1} : \interpret{\Gamma_1} \to \interpret{\tau_1}$ and $\interpret{\mathcal{D}_2} : \interpret{\Gamma_2} \to \interpret{\tau_2}$.
        Since $\otimes$ is a monoidal functor, we compose these two morphisms in parallel:
        \begin{equation*}
            \interpret{\mathcal{D}_1} \otimes \interpret{\mathcal{D}_2} : \interpret{\Gamma_1} \otimes \interpret{\Gamma_2} \to \interpret{\tau_1} \otimes \interpret{\tau_2} = \interpret{\tau_1 \otimes \tau_2}
        \end{equation*}
        \item[Case-Tensor] By the induction hypothesis, there exist morphisms $\interpret{\mathcal{D}_1} : \interpret{\Gamma_1} \to \interpret{\tau_1 \otimes \tau_2}$ and $\interpret{\mathcal{D}_2} : \interpret{\Gamma_2} \otimes \interpret{\tau_1} \otimes \interpret{\tau_2} \to \interpret{\tau_3}$.
        We have:
        \begin{align*}
            \idmorph{\interpret{\Gamma_2}} \otimes \interpret{\mathcal{D}_1} : \interpret{\Gamma_2} \otimes \interpret{\Gamma_1} \to \interpret{\Gamma_2} \otimes \interpret{\tau_1} \otimes \interpret{\tau_2}
        \end{align*}
        and hence by composition:
        \begin{align*}
             \interpret{\mathcal{D}_2}\circ (\idmorph{\interpret{\Gamma_2}} \otimes \interpret{\mathcal{D}_1}) : \interpret{\Gamma_2} \otimes \interpret{\Gamma_1} \to \interpret{\tau_3}
        \end{align*}
        \item[Unit-P] We have to construct a morphism from $1$ to $\interpret{\tyunitp}$, which is defined to be $1$.
        We can therefore pick the identity morphism on $1$.
        \item[Case-Unit-P] By the induction hypothesis there are morphisms $\interpret{\mathcal{D}_1} : \interpret{\Gamma_1} \to 1$ and $\interpret{\mathcal{D}_2} : \interpret{\Gamma_2} \to \interpret{\tau}$.
        We can compose these two morphisms in parallel to obtain $\interpret{\mathcal{D}_1} \otimes \interpret{\mathcal{D}_2} : \interpret{\Gamma_1} \otimes \interpret{\Gamma_2} \to 1 \otimes \interpret{\tau}$.
        There is an isomorphism $\rho$ which witnesses that $1$ is the unit of $\otimes$, so we have that $\rho \circ (\interpret{\mathcal{D}_1} \otimes \interpret{\mathcal{D}_2}) : \interpret{\Gamma_1} \otimes \interpret{\Gamma_2} \to \interpret{\tau}$.
        \item[Abs] We assume there is a derivation $\mathcal{D} : \Gamma, x : \tau_1 \vdash e : \tau_2$.
        By the induction hypothesis there exists a morphism $\interpret{\mathcal{D}} : \interpret{\Gamma} \otimes \interpret{\tau_1} \to \interpret{\tau_2}$.
        Hence $\Lambda \interpret{\mathcal{D}} : \interpret{\Gamma} \to (\interpret{\tau_1} \to \interpret{\tau_2})$.
        \item[App] We assume there are derivations $\mathcal{D}_1 : \Gamma_1 \vdash e_1 : \tau_1 \multimap \tau_2$ and $\mathcal{D}_2 : \Gamma_2 \vdash e_2 : \tau_1$.
        By the induction hypothesis there exist morphisms $\interpret{\mathcal{D}_1} : \interpret{\Gamma_1} \to (\interpret{\tau_1} \to \interpret{\tau_2})$ and $\interpret{\mathcal{D}_2} : \interpret{\Gamma_2} \to \interpret{\tau_1}$.
        By composing them first in parallel and then with the evaluation function:
        \begin{center}
          \begin{tikzcd}
            \interpret{\Gamma_1} \otimes \interpret{\Gamma_2} \ar[rr, "\interpret{\mathcal{D}_1} \otimes \interpret{\mathcal{D}_2}"] & &
            (\interpret{\tau_1} \to \interpret{\tau_2}) \otimes \interpret{\tau_1} \ar[r, "\mathrm{eval}"] &
            \interpret{\tau_2}
          \end{tikzcd}
        \end{center}
    \end{description}
    Lastly, we show how to interpret the rules for the cost monad and the potential type.
    \begin{description}
        \item[Ret] We assume that there is a derivation $\mathcal{D} : \Gamma \vdash e : \tau$.
        By the induction hypothesis there exists a morphism $\interpret{\mathcal{D}} : \interpret{\Gamma} \to \interpret{\tau}$.
        \begin{center}
          \begin{tikzcd}
            \interpret{\Gamma} \ar[r, "\interpret{\mathcal{D}}"] &
            \interpret{\tau} \ar[r, "\eta_{\interpret{\tau}}"] &
            M(0,\interpret{\tau})
          \end{tikzcd}
        \end{center}
        \item[Run] We assume that there is a derivation $\mathcal{D} : \Gamma \vdash e : \costmonad{0}{\tau}$.
        By the induction hypothesis there exists a morphism $\interpret{\mathcal{D}} : \interpret{\Gamma} \to \interpret{\costmonad{0}{\tau}}$.
        \begin{center}
          \begin{tikzcd}
            \interpret{\Gamma} \ar[r, "\interpret{\mathcal{D}}"] &
            M(0,\interpret{\tau}) \ar[r, "\eta^{-1}_{\interpret{\tau}}"] &
            \interpret{\tau}
          \end{tikzcd}
        \end{center}
        \item[Tick] Provided by the additional algebraic operation that comes with
        graded monad:
        \begin{center}
         \begin{tikzcd}
            1 \ar[r, "\uparrow_{p}"] & M(p,1)
          \end{tikzcd}
        \end{center}
        \item[Bind] We assume that there are derivations $\mathcal{D}_1 : \Gamma_1 \vdash e_1 : \costmonad{p_1}{\tau_1}$ and $\mathcal{D}_2 : \Gamma_2, x : \tau_1 \vdash e_2 : \costmonad{p_2}{\tau_2}$.
        By the induction hypothesis we have morphisms $\interpret{\mathcal{D}_1} : \interpret{\Gamma_1} \to M(p_1,\interpret{\tau_1})$ and $\interpret{\mathcal{D}_2} : \interpret{\Gamma_2} \otimes \interpret{\tau_1} \to M(p_2,\interpret{\tau_2})$.
        \begin{center}
          \begin{tikzcd}
          \interpret{\Gamma_1} \otimes \interpret{\Gamma_2} \ar[rr, "\sigma_{\interpret{\Gamma_1},\interpret{\Gamma_2}}"] &&
          \interpret{\Gamma_2} \otimes \interpret{\Gamma_1} \ar[rr, "\idmorph{\interpret{\Gamma_2}} \otimes \interpret{\mathcal{D}_1}"] &&
          \interpret{\Gamma_2} \otimes M(p_1, \interpret{\tau_1}) \ar[dllll, "\mathit{st}_{\interpret{\Gamma_2}, \interpret{\tau_1}}", swap] \\
          M(p_1, \interpret{\Gamma_2} \otimes \interpret{\tau_1}) \ar[rr, "{M(\idmorph{p_1},\interpret{\mathcal{D}_2})}", swap] &&
          M(p_1, M(p_2, \interpret{\tau_2})) \ar[rr, "\mu_{p_1,p_2}", swap] &&
          M(p_1 + p_2, \interpret{\tau_2})
          \end{tikzcd}
        \end{center}

        \item[Store] We assume there is a typing derivation $\mathcal{D} : \Gamma \vdash e : \tau$.
        By the induction hypothesis there exists a morphism $\interpret{\mathcal{D}} : \interpret{\Gamma} \to \interpret{\tau}$.
        We have to construct a morphism from $\interpret{\Gamma}$ to $M(p, P(p,\interpret{\tau}))$.
        Let us look at one specific instance of the adjunction between $M$ and $P$:
        \begin{equation*}
          \phi : \homset{\modelcategory}{P(p,\interpret{\Gamma})}{P(p,\interpret{\tau})} \simeq \homset{\modelcategory}{\interpret{\Gamma}}{M(p,P(p,\interpret{\tau}))} : \phi^{-1}
        \end{equation*}
        We can thus pick $\phi(P(\idmorph{p}, \interpret{\mathcal{D}}))$ as a morphism from $\interpret{\Gamma}$ to $M(p,P(p,\interpret{\tau}))$.


        \item[Pay] We assume there is a typing derivation $\mathcal{D} : \Gamma \vdash e : \potentialtype{p}{\costmonad{p}{\tau}}$.
        By the induction hypothesis there exists a morphism $\interpret{\mathcal{D}} : \interpret{\Gamma} \to \interpret{\potentialtype{p}{\costmonad{p}{\tau}}}$.
        We can compose this morphism with the counit of the adjunction to obtain the morphism $\epsilon \circ \interpret{\mathcal{D}} : \interpret{\Gamma} \to \interpret{\tau}$.

        \item[Split] We assume there is a typing derivation $\mathcal{D} : \Gamma \vdash e : \costmonad{(p_1 + p_2)}{\tau}$.
        By the induction hypothesis there exists a morphism $\interpret{\mathcal{D}} : \interpret{\Gamma} \to \interpret{\costmonad{(p_1 + p_2)}{\tau}}$.
        We can compose this morphism with the inverse of the multiplication of the graded monad $M$ to obtain
        $\mu^{-1} \circ \interpret{\mathcal{D}} : \interpret{\Gamma} \to \interpret{\costmonad{p_1}{(\costmonad{p_2}{\tau})}}$.

        \item[PLet] We assume there are typing derivations $\mathcal{D}_1 : \Gamma_1 \vdash e_1 : \potentialtype{p}{\tau_1}$ and $\mathcal{D}_2 : \Gamma_2, x : \tau_1 \vdash e_2 : \tau_2$.
        By the induction hypothesis there exist morphisms $\interpret{\mathcal{D}_1} : \interpret{\Gamma_1} \to \interpret{\potentialtype{p}{\tau_1}}$ and $\interpret{\mathcal{D}_2} : \interpret{\Gamma_2} \otimes \interpret{\tau_1} \to \interpret{\tau_2}$.
        \begin{center}
          \begin{tikzcd}
            \interpret{\Gamma_1} \otimes \interpret{\Gamma_2} \ar[rr, "\sigma_{\interpret{\Gamma_1},\interpret{\Gamma_2}}"] &&
            \interpret{\Gamma_2} \otimes \interpret{\Gamma_1} \ar[rr, "\idmorph{\interpret{\Gamma_2}} \otimes \interpret{\mathcal{D}_1}"] &&
            \interpret{\Gamma_2} \otimes P(p,\interpret{\tau_1}) \ar[dllll, "\mathit{st}_{\interpret{\Gamma_2},\interpret{\tau_1}}", swap] \\
            P(p, \interpret{\Gamma_2} \otimes \interpret{\tau_1}) \ar[rrrr, "{P(p, \interpret{\mathcal{D}_2})}",swap] &&&&
            P(p, \interpret{\tau_2})
          \end{tikzcd}
        \end{center}

    \end{description}
\end{proof}

We also introduced substitution of typing derivations as a lemma in \cref{def:lamor:substitution-derivations}.
These are interpreted as follows:
\begin{lemma}[Interpretation of Substitution]
  If $\mathcal{D}_1 : \Gamma_1 \vdash e_1 : \tau_1$ and $\mathcal{D}_2 : \Gamma_2,x : \tau_1 \vdash e_2 : \tau_2$,
  then \cref{thm:adjoint-model:fundamental} gives us morphisms $\interpret{\mathcal{D}_1} : \interpret{\Gamma_1} \to \interpret{\tau_1}$ and
  $\interpret{\mathcal{D}_2} : \interpret{\Gamma_2} \otimes \interpret{\tau_1} \to \interpret{\tau_2}$.
  We can combine them as
  \begin{align*}
    \interpret{\mathcal{D}_2} \circ (\idmorph{\interpret{\Gamma_2}} \otimes \interpret{\mathcal{D}_1}) \circ \sigma_{\interpret{\mathcal{D}_1},\interpret{\mathcal{D}_2}}: \interpret{\Gamma_1} \otimes \interpret{\Gamma_2} \to \interpret{\tau_2}
  \end{align*}
  which agrees with the interpretation of $\textsc{Subst}_1(\mathcal{D}_1,\mathcal{D}_2)$ of \cref{def:lamor:substitution-derivations}.
  Similarly for $\textsc{Subst}_2$.
\end{lemma}

\begin{remark}[Deriving tick from subtyping]
\label{rem:tick-from_sub}
For cost models where $0 \leq p$ for all $p$ (i.e., $0$ is initial in the
$\costcategory{}$), then we can derive the tick operation $\uparrow_p : 1 \to
M(p, 1)$ from the model of subtyping:
         \begin{center}
          \begin{tikzcd}
            1 \ar[r, "\eta_{1}"] & M(0,1) \ar[r, "\interpret{\mathcal{S}}"] & M(p,1)
          \end{tikzcd}
         \end{center}
We do not bake in this choice however to avoid the need for $0$ to be the bottom element
and to decouple the notion of enforcing a cost (which may want a computational interpretation)
from the model of subtyping (which may \emph{not} want a computational interpretation).
\end{remark}


\begin{theorem}[Soundness]
    If $\mathcal{D}_1 \equiv \mathcal{D}_2$, then $\interpret{\mathcal{D}_1} = \interpret{\mathcal{D}_2}$.
\end{theorem}

\begin{proof}
Since the model is close to structure of derivations, the proof is straightforward, in some cases with an exact 1-1 correspondence.
\end{proof}

\subsection{Set Model}
\label{sec:set-model}

The simplest instance of the model presented in the previous section is the set-theoretic model which ignores the type-theoretic analysis of cost and potential and maps the monad and potential types directly to the representation of their underlying types:

\begin{definition}[Set-Theoretic Model]
    The model category $\modelcategory$ is just $\setcategory$, the category of sets.
    Both monoidal structures $(\modelcategory, \times, 1, \alpha, \lambda, \rho, \sigma)$ and $(\modelcategory, \otimes, 1, \alpha, \lambda, \rho,\sigma)$ are given by the categorical product in $\setcategory$.
    The monoidal structure $(\modelcategory, +, 0, \alpha, \lambda, \rho, \sigma)$ is given by the categorical coproduct in $\setcategory$.
    Both functors $M : \costcategory \times \setcategory \to \setcategory$ and $P : \costcategory^\op \times \setcategory \to \setcategory$ are given by identity functors which ignore the category $\costcategory$ completely.
\end{definition}

It can easily be checked that this definition satisfies all the required monad and functor laws, including strength, and that there is an adjunction between the two identity functors.

\section{$\lambda$-amor's Kripke Logical Relation Model}
\label{sec:kripke-model}
Previous work defined a logical relation model of
$\lambda$-amor~\cite{Rajani2021lamor}, shown
below. We leverage this to form an instance of the $\lambda$-amor Core
model of Section~\ref{sec:adjoint-model}.

Section~\ref{sec:lamor-cat-model} argued for a reduced set of language
primitives from the perspective of a categorical model for
$\lambda$-amor centered around a family of adjunctions between
potential and cost, and a structure for cost that is both graded
monadic and graded comonadic.  To justify that this is a rational
reconstruction of $\lambda$-amor, we show in this section that from
the original logical relations model for $\lambda$-amor we can form a
sound \emph{syntactic model} for $\lambda$-amor Core. A further
implication of this result is that from $\lambda$-amor primitives we
can derive the $\lambda$-amor Core primitives.

The model is defined in terms of a value interpretation on
types $\kripkemodel{\tau}$, an interpretation of contexts
$\kripkemodel{\Gamma}$, and an expression-based interpretation of types
$\kripkemodele{\tau}$:


\begin{definition}[Kripke Logical Relation Model]
  \label{def:kripke-model:model}
  \begin{displaymath}
    \begin{array}{llll}
      \kripkemodel{\tyunitp} & \coloneq & \{(p, \tmunit)\} & \\
      \kripkemodel{\tyunitn} & \coloneq & \{ (p, \tmcocase{})\} & \\
      \kripkemodel{\tyvoid}  & \coloneq & \emptyset & \\
      \kripkemodel{\tau_1 \oplus \tau_2} & \coloneq & \lbrace (p,\tminl{v}) \mid (p,v) \in \kripkemodel{\tau_1} \rbrace \cup \lbrace (p, \tminr{v}) \mid (p,v) \in \kripkemodel{\tau_2} \rbrace & \\
      \kripkemodel{\tau_1 \otimes \tau_2} & \coloneq & \{ (p, \tmtensor{v_1}{v_2}) \mid \exists \potential_1, \potential_2. \potential_1 + \potential_2 \leq p \wedge (p_1, v_1) \in \kripkemodel{\tau_1} \wedge (p_2,v_2) \in \kripkemodel{\tau_2} \} & \\
      \kripkemodel{\tau_1 \with \tau_2} & \coloneq & \{ (p, \tmcocase{\tmfst \Rightarrow t_1; \tmsnd \Rightarrow t_2}) \mid (p, t_1) \in \kripkemodele{\tau_1} \wedge (p,t_2) \in \kripkemodele{\tau_2} \} & \\
      \kripkemodel{\tau_1 \multimap \tau_2}& \coloneq & \{ (p,\lambda x. e) \mid \forall p',v . (p', v) \in \kripkemodele{\tau_1} \Rightarrow (p + p', e[v/x]) \in \kripkemodele{\tau_2} \} \\
      \kripkemodel{\costmonad{\cost}{\tau}} & \coloneq & \{ (p,v) \mid \exists \potential' . \potential' + \cost' \leq \potential + \cost \wedge v \forcesto{\cost'} v' \wedge
                                                          (p',v') \in \kripkemodel{\tau} \} & \\
      \kripkemodel{\potentialtype{n}{\tau}}& \coloneq & \{ (p,v) \mid \exists \potential'. n + \potential' \leq \potential \wedge (\potential',v) \in \kripkemodel{\tau} \} &
      \\[0.5em]
      \kripkemodel{\Gamma} & \coloneq & \{(p,\gamma) \mid \exists f : \mathsf{dom}(\Gamma) \rightarrow \costcategory . \\
      & & (\mathcal\forall x \in \mathsf{dom}(\Gamma) . (f(x), \gamma(x)) \in \kripkemodel{\Gamma(x)}) \wedge (\sum_{x \in \mathsf{dom}(\Gamma)} f(x)) \leq p \} \\[0.5em]
      \kripkemodele{\tau} & \coloneq & \{ (p,e) \mid \forall v . e \evaluatesto v \wedge (p,v) \in \kripkemodel{\tau} \}
    \end{array}
  \end{displaymath}
\end{definition}

\noindent
The interpretation of the cost monad is that: given some concrete amount of fuel
$\kappa'$ used-up by forcing the term $v$ and given that we require fuel $p'$ to model
the inner term (type $\tau$) then overall $\kappa' + p'$ must be no more than
the incoming fuel $p$ plus the extra fuel $\kappa$ needed by this term, as recorded in the
type. For the potential type, since we are storing at most $n$ fuel here then
the amount of fuel available $p$ must be at least $n$ plus the amount $p'$ required by
the inner term of type $\tau$.

We view the model in the form $\kripkemodel{\tau}_p$ as cost-indexed
  set of terms (hence the \emph{Kripke} naming):
  $$
  \kripkemodel{\tau}_p \triangleq \bigcup \{ t \mid (p, t) \in \kripkemodel{\tau} \}
  $$
  Whenever $p \leq p'$ then $\kripkemodel{\tau}_p \subseteq \kripkemodel{\tau}_{p'}$.
This model induces the structure of a category.
\begin{definition}[Kripke Category, \lean]
  \label{def:kripke-model:category}
   For all cost categories $\costcategory$ in which $0$ is initial,
 i.e., that $0 \leq p$ for all $p \in \costcategory$, then
  we can define the category $\kripkecategory$ by the following data:
  \begin{itemize}
    \item Objects of $\kripkecategory$ are types.
    \item For any $A,B \in \catobjects{\kripkecategory}$, a morphism from $A$ to $B$ is an $\costcategory{}$-indexed set of terms $t_1$ such that, for every $p \in \costcategory{}$, $t_1$ is a term under exactly one binder (i.e., $\lambda x . t_1$) and:
\begin{align*}
        \forall q, t_2 ,\ t_2 \in \kripkemodel{A}_q \implies t_1[t_2/0] \in \kripkemodel{B}_{p + q}
    \end{align*}
    That is, for all terms $t_2$ with cost $q$ in the Kripke interpretation of the
    source then the substitution of that term into the body $t_1$ is in the model of the target at cost $p + q$.
%

\end{itemize}
\end{definition}

We can prove that this category is an instance of the generic model
of Section~\ref{sec:adjoint-model}, defining the rest of the structure
atop of $\kripkecategory$:

\begin{theorem}[Model Instance, \lean]
\label{thrm:kripke-model-instance}
The Kripke
 category $\kripkecategory$ defined
 in \cref{def:kripke-model:category} provides an instance of the
 generic model defined in \cref{def:adjoint-model:model-category}.
\end{theorem}

\begin{proof}
We outline the main ideas of the construction and proof, which has been mechanised
(proofs scripts accompanying this paper). At various points, the initiality of $0$
is required.

\begin{enumerate}
    \item Products given by $\tau_1 \& \tau_2$ which is
    a bifunctor with morphism mapping such that for all $f : \tau \to \tau'$ and $g : \sigma \to \sigma'$ then $f \, \& \, g \triangleq \lambda x . \tmcocase{\tmfst \Rightarrow
    x.\tmfst, \tmsnd \Rightarrow x.\tmsnd}$. The corresponding unit is given by
    $\top$ which is terminal
    $!_\tau \triangleq (\lambda x . \tmcocase{}) : \tau \to \top$.

    \item Symmetric monoidal structure $\tau_1 + \tau_2$ which is a bifunctor
    with morphism mapping $f + g \triangleq \lambda x . \tmcase{x}{\tminl{y} \Rightarrow
    \tminl{y}, \tminr{z} \Rightarrow \tminr{z}}$ with unit $0$.

    \item Symmetric monoidal structure $\tau_1 \otimes \tau_2$ which is a bifunctor with morphism mapping $f \otimes g \triangleq
    \lambda x . \tmcase{x}{\tmtensor{y}{z} \Rightarrow \tmtensor{f y}{g z}}$
    with unit $1$ which is terminal with
    $!_\tau \triangleq (\lambda x . \tmunit) : \tau \to 1$.

    \item The right adjoint to $\otimes$ given
     by functions $\tau_1 \multimap \tau_2$ with the isomorphism of hom-sets by:
    \setlength{\arraycolsep}{0.1em}
    \begin{align*}
    \begin{array}{rll}
    \Phi_\rightarrow (f : (\sigma \otimes \tau_1) \multimap \tau_2) & \triangleq \lambda z . \lambda y . f \tmtensor{z}{y} & : \sigma \multimap (\tau_1 \multimap \tau_2) \\
    \Phi_\rightarrow^{-1} (g : \sigma \multimap (\tau_1 \multimap \tau_2)) &
    \triangleq \lambda x . \tmcase{x}{\tmtensor{y}{z} \Rightarrow g y z}
    & : (\sigma \otimes \tau_1) \multimap \tau_2
    \end{array}
    \end{align*}

    \item A \emph{strong graded monad} structure on $\costmonad{p}{-}$
    with a $\costcategory$-indexed family of
    strong endofunctors on $\kripkecategory$ with the morphism mappings defined such
    that for all $f \in \tau_1 \to \tau_2$ then:
    \[
    \costmonad{p}{f} \triangleq (\lambda x. \tmbind{y}{x}{\tmret{(f\ y)}}) :
    \costmonad{p}{\tau_1} \to \costmonad{p}{\tau_2}
    \]
    and with strong graded monad operations:
    \setlength{\arraycolsep}{0.1em}
    \begin{align*}
    \begin{array}{rll}
     \eta_\tau & \triangleq \lambda x . \tmret{x} & : \tau \to \costmonad{0}{\tau} \\
     \mu_{p_1,p_2,\tau} & = \lambda x. \tmbind{y}{x}{y}
     & : \costmonad{p_1}{(\costmonad{p_2}{\tau})} \multimap \costmonad{(p_1 + p_2)}{\tau} \\
    \mathit{st}_{p,\tau_1,\tau_2} & \triangleq
   \lambda x . \tmcase{x}{\tmtensor{y}{z} \Rightarrow \tmbind{z'}{z}{\tmret{\tmtensor{y}{z'}}}} & : (\tau_1 \otimes \costmonad{p}{\tau_2}) \to \costmonad{p}{(\tau_1 \otimes \tau_2)}
     \end{array}
     \end{align*}

    \item The $\uparrow_p$ tick operation is given by the graded monad's unit operation and
    by our assumption that $0$ is initial thus we can use subtyping (see Remark~\ref{rem:tick-from_sub}),
    i.e., $\uparrow_p \triangleq \lambda x . \tmret{x} : 1 \to \costmonad{p}{1}$

    \item A \emph{graded comonad} structure on $\costmonad{p}{-}$, that is compatible
    with the graded monad (i.e., operations form an isomorphism), with operations:
    \begin{align*}
    \begin{array}{rll}
    \varepsilon_\tau & \triangleq \lambda x . \tmrun{x} & : \costmonad{0}{\tau} \to \tau \\
    \delta_{p_1,p_2,\tau} & \triangleq \lambda x . \tmbind{y}{\tmstore{x}}{\tmret{\tmrelease{z}{y}{z}}}
    & : \costmonad{(p_1 + p_2)}{\tau} \to \costmonad{p_1}{\costmonad{p_2, \tau}}
    \end{array}
    \end{align*}
    following the derivations for $\mathtt{split}$ for $\delta$
    from Section~\ref{subsubsec:split-via-release}

\item The graded monad and comonad operations are compatible, i.e.,
$\eta^{-1} = \varepsilon$ and $\mu^{-1} = \delta$ which follows from
their definitions.

\item The indexed family of strong functors $\potentialtype{p}{-}$ with
morphism mapping for all $f \in \tau_1 \to \tau_2$ and strength given by:
\begin{align*}
\begin{array}{rll}
\potentialtype{p}{f} & \triangleq \lambda x . \tmrun{\tmrelease{y}{x}{\tmstore{f\, x}}}
& : \potentialtype{p}{\tau_1} \to \potentialtype{p}{\tau_2} \\
\mathit{st}_{p,\tau_1,\tau_2} & \triangleq
\lambda x . \tmcase{x}{\tmtensor{y}{z} \Rightarrow \tmrun{\tmrelease{z'}{z}{\tmstore{\tmtensor{y}{z'}}}}} & : \tau \otimes \potentialtype{p}{\tau'}
\to \potentialtype{p}{(\tau \otimes \tau')}
\end{array}
\end{align*}

\item A \emph{family of adjunctions} $\potentialtype{p}{-} \dashv \costmonad{p}{-}$ for all $p$, given by the unit and counit operations:
    \begin{align*}
    \begin{array}{rll}
    \varepsilon_{p,\tau} & = \lambda x . \tmrun{\tmrelease{y}{x}{y}} & \in \potentialtype{p}{\costmonad{p}{\tau}} \to \tau \\
    \eta_{p,\tau} & = \lambda x . \tmstore{x} & \in \tau \to \costmonad{p}{\potentialtype{p}{\tau}}
    \end{array}
    \end{align*}
following the derivation for $\mathtt{pay}$ for $\varepsilon$ given
in Section~\ref{subsubsec:pay-via-release}.

\end{enumerate}
Items (3-9) have been fully mechanized in Lean; (1-2) are straightforward and
less the focus here.
\end{proof}

\section{A general Copresheaf Model}
\label{sec:model}
We now consider a model construction based around standard constructions in category
theory which gives a family of models for our
calculus. We take inspiration from two ideas.

Firstly, Section~\ref{sec:kripke-model} gave an instance of the model
reusing the Kripke logical relations model of $\lambda$-amor, built around
cost-indexed set of terms, i.e., a
\emph{(co)presheaf} model on the category
$[\costcategory, \setcategory]$.  (Co)presheaf models
abound in the programming language semantics literature whenever our
semantics depends on some notion of world, context~\cite{hofmann1999semantical},
or resource~\cite{pym2004possible}.
The standard construction of \emph{Day's convolution} \cite{Day1970thesis} then gives a
way to `push down' a combinational theory on contexts/resources into the
semantics; we leverage this approach here.

Secondly, the adjoint structure of cost \emph{vs.} potential that
we exposed in Section~\ref{sec:adjoint-model} suggests a
model using the standard notion of products as left adjoint to exponents,
i.e., $- \otimes R \dashv R \multimap -$, modelling potential via a pair
of a computational result and a resource $R$, and cost as a function consuming
a resource $R$ to produce a result.

We combine these two ideas to give a family of models, comparing it to
model of Section~\ref{sec:model}.  We start with some standard
categorical definitions in order to introduce the notation that we use
for various constructions.  We reuse the definition of the
cost category $\costcategory$ of Definition~{def:adjoint-model:costcategory}.

\begin{remark}
Consider taking the model category to be that of copresheaves over $\costcategory$,
i.e., $\modelcategory = [\costcategory, \setcategory]$.
    The intuition here is that the interpretation of any type $\tau$ is a functor from $\costcategory$ to Set, and can therefore be \enquote{applied} to different costs.
    For example, $\interpret{\tau}(0)$ contains the interpretation of all terms of type $\tau$ that can be evaluated with cost 0, $\interpret{\tau}(5)$ contains those that can be evaluated with cost at most $5$ and so on.
    Since the category $\costcategory$ is ordered there is a (unique, since the category is thin) morphism $f_{0\leq 5}$ from 0 to 5, so $\interpret{\tau}(f_{0\leq 5}) : \interpret{\tau}(0) \to \interpret{\tau}(5)$ witnesses that everything that can be had for free can also be had for the cost of 5. Hence the use of \emph{co}presheaves rather than presheaves.
\end{remark}

Whilst (co)presheaves are traditionally over $\setcategory$, we generalise
here for our model category:

\begin{definition}[Model category $\copresheafmodel$]
    \label{def:copresheaf-model:model}
    Let $\modelcategory$ be a symmetric monoidal closed category which is
    complete and cocomplete (i.e., has all limit and colimits).
    If $\costcategory$ is $\modelcategory$-enriched (i.e., its hom-objects
    are objects of $\modelcategory)$, then
    our model category is given by
    $\copresheafmodel = [\costcategory, \modelcategory]$,
    i.e. co-presheaves over $\costcategory$.
\end{definition}

\begin{definition}[Monoidal product via Day convolution]
\label{def:day-conv}
The copresheaf category $\copresheafmodel$ is
symmetric monoidal with unit $\dayidentity$ and tensor $\dayconvolution$ defined as follows for all $X, Y : [\costcategory, \modelcategory]$
\begin{align*}
(X \dayconvolution Y)(p) & \triangleq \int^{p_1,p_2} X(p_1) \otimes Y(p_2) \otimes \costcategory(p_1 + p_2, p) \\
\dayidentity(p) & \triangleq \costcategory(0, p)
\end{align*}
When $\modelcategory = \setcategory$, then the above coend specialises to:
$\exists p_1, p_2 . \int^{p_1,p_2} X(p_1) \times Y(p_2) \times \costcategory(p_1 + p_2, p)$
and since $\costcategory$ is thin then $\costcategory(p_1 + p_2, p)$ is either
a singleton set or empty.
\end{definition}

\begin{definition}[Day internal hom]
\label{def:day-exponent}
Further, the copresheaf category $\copresheafmodel$ is symmetric
monoidal closed with the addition of the \emph{Day internal hom}, which is
right adjoint to the Day tensor:
\[
(X \dayexponential Y)(p) \triangleq \int_{p_1, p_2}
\costcategory(p + p_1, p_2) \multimap (X(p_1) \multimap Y(p_2))
\]
\end{definition}

A further useful construction used below is that of the \emph{coYoneda embedding}:
\begin{definition}
For a $\modelcategory$-enriched category $\costcategory$, then the
\emph{coYoneda} embedding $\yoneda{-} : \costcategory^\op \to
[\costcategory, \modelcategory]$ is given by the hom-functor:
\begin{equation}
\yoneda{n} = \costcategory(n, -)
\end{equation}
The Day unit thus can be re-expressed as $\dayidentity \triangleq \yoneda{0}$.

An implication of this definition is also that, for all $n, m \in \costcategory$,
we have:
\begin{equation}
\label{eq:yoneda-comb}
\yoneda{n} \dayconvolution \yoneda{m} \cong \yoneda{n + m}
\end{equation}
\end{definition}

\begin{theorem}[Model Instance]
    The category $\copresheafmodel$ defined in \cref{def:copresheaf-model:model} is an instance of the generic adjoint model defined in  \cref{def:adjoint-model:model-category}.
\end{theorem}
\begin{proof}
    We go through the properties required in \cref{def:adjoint-model:model-category} one-by-one.
    \begin{enumerate}
        \item For the interpretation of $\&$ and $\top$, the categorical product $\times$ is formed on $\copresheafmodel$ via
        the pointwise construction i.e., $(X \times Y)(p) = X(p) \times Y(p)$ and
        similarly for the terminal object $1(p) \triangleq 1$;

        \item The symmetric monoidal structure for $+$ and $0$ is also given by the standard pointwise construction in copresheaves;

        \item The symmetric monoidal structure for $\otimes$ and $1$ is given by Day convolution $\dayconvolution$ and its unit $\dayidentity$ (Definition~\ref{def:day-conv});

        \item $\dayconvolution$ is left adjoint to $\dayexponential$.

        \item The strong graded monad $M : \costcategory \times \copresheafmodel \to \copresheafmodel$ is given by
        $$M(n,X) = \yoneda{n} \dayexponential X$$
        This functor is covariant in both arguments. The monad operations are then essentially those
        of the standard \emph{reader monad}.

        \item The $\uparrow_p$ operation corresponds here to one of type
        $\dayidentity \to \yoneda{p} \dayexponential \dayidentity$.
        Essentially we then require that there exists in the underlying model category
        $\modelcategory{}$ a family of maps $\yoneda{p}(q) \multimap \yoneda{0}(q)$
        for all $p$ and $q$. Recalling Remark~\ref{rem:tick-from_sub}, this is always
        provided when $0$ is the least element: in this case,
        since the coYoneda embedding is contravariant in its first argument then
        from $0 \leq p$ we get $\yoneda{p} \dayexponential \yoneda{0}$ by functoriality.

        \item The graded monad construction also admits a graded comonad,
        essentially the \emph{graded product comonad}.

        \item Furthermore, by Equation~\eqref{eq:yoneda-comb} we have that:
        \begin{align*}
        M(p, M(q, X)) \cong M(p + q, X)  \qquad M(0, X) \cong X
        \end{align*}
        and thus we have both a graded monad and graded comonad that are
        compatible.

        \item The strong functor $P : \costcategory^\op \times \copresheafmodel \to \copresheafmodel$ is given by $$P(n,X) = X \dayconvolution \yoneda{n}$$
              This functor is contravariant in its first argument, and covariant in its second argument.

        \item For every $p \in \catobjects{\costcategory}$ there is an adjunction
        given by the adjunction of $(- \dayconvolution \yoneda{p}) \dashv (\yoneda{p} \dayexponential -)$.
    \end{enumerate}
\end{proof}


\begin{remark}[Interpretation of the potential type]
Expanding the definition of the model of the indexed potential type, at some
cost $p$, we get:
\begin{align*}
\interpret{\potentialtype{n}{\tau}}(p)
& = \hspace{-0.3em} \int^{p_1,p_2} \interpret{\tau}(p_1) \otimes \yoneda{n}(p_2) \otimes \costcategory(p_1 + p_2, p) \\
& = \hspace{-0.3em} \int^{p_1,p_2} \interpret{\tau}(p_1) \otimes \costcategory(n,p_2) \otimes \costcategory(p_1 + p_2, p) \cong \hspace{-0.3em} \int^{p'} \interpret{\tau}(p') \otimes \costcategory(p' + n, p)
\end{align*}
The isomorphism here simplifies the coend term by the enriched (co)Yoneda lemma.

In $\setcategory$, this specialises to $\exists p', \interpret{\tau}(p') \wedge
p' + n \leq p$. This is exactly the interpretation of the potential type
given in the Kripke logical relations model (Definition~\ref{def:kripke-model:model}).
\end{remark}

For potential types, the copresheaf model coincides with the original Kripke
logical relations model. However, this happy coincidence does not hold
for the cost graded monad:

\begin{remark}[Interpretation of cost graded monad]
Expanding the definition of the model of the graded monad, at some cost $p$,
we get:
\begin{align*}
(\interpret{\costmonad{n}{\tau}})(p) = (\yoneda{n} \dayexponential \interpret{\tau})(p) & = \int_{p_1,p_2} \mathbb{C}(p + p_1, p_2) \multimap  (\yoneda{n}(p_1) \multimap \interpret{\tau}(p_2)) \\ 
%
& \cong \int_{p'} \mathbb{C}(n, p') \multimap \interpret{\tau}(p + p')
\end{align*}
In $\setcategory$, this specialises to
$\forall p', (n \leq p') \Rightarrow \interpret{\tau} (p + p')$, which
differs considerably to the model of the cost graded monad in the logical
relations model which incorporates also the forcing relation.
Instead, we can read the above model as saying that given some incoming fuel
$p$ then we need some further fuel $p'$ that has to be at least as big as $n$
in order to satisfy the cost demands of $\tau$.
\end{remark}
Thus, the copresheaf model and Kripke logical relation model give different
models for cost.

We conclude this section with an observation that shows we can run a computation if we provide enough fuel, which is here represented by the coYoneda embedding of the corresponding cost.

\begin{theorem}[Enough Fuel]
    \label{thm:model:enough-fuel}
    If $p \leq p'$, then there is a morphism:
    \begin{align*}
        \interpret{\costmonad{p}{\tau}} \dayconvolution \yoneda{p'} \to \interpret{\tau}
    \end{align*}
\end{theorem}
\begin{proof}
   From $p \leq p'$ we obtain a morphism $f : \yoneda{p'} \to \yoneda{p}$.
   If we unfold the definition, then $\interpret{\costmonad{p}{\tau}} \dayconvolution \yoneda{p'}$ is the same as $(\yoneda{p} \dayexponential \interpret{\tau}) \dayconvolution \yoneda{p'}$.
   We can thus define:
   \begin{equation*}
      \textrm{Day-Apply} \circ (\idmorph{\costmonad{p}{\tau}} \dayconvolution f) : \interpret{\costmonad{p}{\tau}} \dayconvolution \yoneda{p'} \to \interpret{\tau}
   \end{equation*}
\end{proof}

\section{Related Work}
\label{sec:related-work}
\paragraph{Type-theoretic amortized cost analysis}
Verification of amortized cost bounds using types relies on adaptations of Tarjan's
potential method~\cite{Tarjan1985amortized} into a type-theoretic framework, which tracks resource
consumption by assigning imaginary "potential" to types. Early work by \citet{DBLP:conf/popl/HofmannJ03} embodied this by building a type and effect system to track
potentials and using that to account for cost of programs steps. This methodology was subsequently
extended to more a more general setting through Resource-Aware ML (RAML) by \citet{DBLP:journals/toplas/0002AH12}, utilising multivariate polynomial potential to capture
complex, non-linear execution bounds. Besides such type and effect approaches, graded modal types
have also been used to study amortized cost bounds. The main work in that space is the
$\lambda_{\text{amor}}$ framework from \citet{Rajani2021lamor} which is also the main object of
study of this work.
Also worth mentioning is the work by \citet{Danielsson2008semiformal} which introduced a lightweight and semiformal approach
using a type system (mechanized in Agda) that wraps computations in a graded \texttt{Thunk} monad,
tracking exact evaluation steps directly within the program's type index to analyze purely
functional, persistent data structures. Seeking a more mathematical bridge between program execution
and traditional algorithmic analysis, \citet{DBLP:journals/pacmpl/KavvosMLD20} utilized
Call-by-Push-Value (CBPV) to formalize a uniform framework for recurrence extraction across both
call-by-value and call-by-name strategies. This denotational extraction methodology was later
extended by \citet{DBLP:journals/pacmpl/CutlerLD20} to accommodate amortized
analysis by establishing a formal bounding logical relation over a source language augmented with
banker's style time credits.
Finally, approaches like those by \citet{Niu2022calf,Grodin2024} utilize
dependent type theory to structurally separate cost tracking from functional correctness.

\paragraph{Graded monads}
Graded monads provide a generalisation of monads which can be used to provide a precise model of
effect systems~\cite{Katsumata2014parametric,Orchard2014marriage,DBLP:conf/birthday/MycroftOP16},
dataflow analyses of stateful programs~\cite{ivavskovic2022programming}, and in probabilistic
programming~\cite{perrone2018categorical}. The application of graded monads to cost has appeared in
the type system of Granule~\cite{DBLP:journals/pacmpl/OrchardLE19} and in the context of modelling
Graded Hoare Logic~\cite{Gaboardi2021gradedhoare}. The line of work by~\citet{Rajani2021lamor} and
\citet{Rajani2024modal} built from this view of graded cost to the amortized view as described in
this paper.

\paragraph{Other presheaf models}

\citet{Fiore1999abstract} develop a model for variable binding that is
close to our approach here, based on copresheaves $\mathcal{F} = \mathbf{Set}^{\mathbb{F}}$
where $\mathbb{F}$ is category whose objects are finite sets $\{1, \ldots, n\}$
(for some $n \in \mathbb{N}$) and functions between such sets. Copresheaves
 $\mathcal{F}$ are used as a model of types, indexed
(or stratified) by the \emph{size} of the free-variable context. A (co)presheaf
of \emph{abstract variables} $V \in \mathcal{F}$ is defined by the Yoneda embedding
to present a context with one variable:
$$
V(n) = \mathcal{F}(1, n) =n.
$$
A context extension operator $\delta : \mathcal{F} \to \mathcal{F}$ captures
the idea that $\delta A$ in a context $n$ is an element of $A$ in an extended context $n + 1$:
$$
(\delta A) n = A(n + 1)
$$
They note that $\delta$ is equal to the cartesian exponential $[V,-]$ and is therefore right adjoint to $(-)\times V$:
$$
\mathcal{F}(X, \delta A) \cong \mathcal{F}(X, [V, A]) \cong \mathcal{F}(X \times V, A)
$$
In other words, this parallels our construction with $\delta A = \yoneda{1} \dayexponential A =  \interpret{\costmonad{1}{\tau}}$ (assuming $A =\interpret{\tau}$) and
$A \times V = A \dayconvolution \yoneda{1} = \interpret{\potentialtype{1}{\tau}}$; thus
we can see the size of a context as a kind of cost.

\citet{Biering2004logic} defines a Kripke-Joyal semantics for bunched implications
in a presheaf over a small symmetric monoidal category $(\mathcal{C}, \cdot, e)$
(\cite[Theorem 6.3.3]{Biering2004logic}). Similarly to our tensors (modulo their use of contravariance and our use of covariance), the model of separating conjunction $\ast$ is then Day convolution: for some $c \in C$ and propositions $p, q$ then $c \models p \ast q\ \iff\ \int^{d,d' \in C} C(c, d \cdot d') \times (d \models p) \times (d' \models q)$ (by a slight abuse of (our) notation to draw out the similarity), and the model of implication is given by the Day exponent.


Where we have externalized the cost monoid as the base of a graded copresheaf model, \citet{Grodin2024} account for cost as a writer-monad effect $X \mapsto X \times C$ internal to a single (effectful) type theory.



\section{Future Work}
\label{sec:future-work}
One of properties that we expect to hold for the class of adjoint models we have defined in \cref{sec:adjoint-model} is the following completeness theorem.

\begin{theorem}[Completeness]
    Let $\mathcal{D}_1 : \Gamma \vdash e_1 : \tau$ and $\mathcal{D}_2 : \Gamma \vdash e_2 : \tau$.
    There exists a model such that if $\interpret{\mathcal{D}_1} = \interpret{\mathcal{D}_2}$ in that model, then $\mathcal{D}_1 \equiv \mathcal{D}_2$.
\end{theorem}

Apart from this missing theorem we see the following ways to extend our work.
We want to extend the categorical semantics to the complete theory of $\lambda$-amor introduced in \citet{Rajani2021lamor}, in particular, this means that we also have to cover recursion and exponentials.
We also plan to extend the Lean formalization to cover the model introduced in \cref{sec:model} and to fully formalize \cref{thm:adjoint-model:fundamental}.

\section{Conclusion}
\label{sec:conclusion}
It is well-established today that type systems can be used to soundly reason about the amortized cost of algorithms.
What was lacking, however, were categorical (and therefore compositional) models that can explain the meaning of the graded cost (co-)monad and its interaction with the potential type --- interactions which are essential for understanding how amortized cost analysis works.
Our paper provides this missing piece of the puzzle: We showed how to model the relationship between cost and potential using a family of adjunctions, and presented three instances of the general scheme.
Furthermore, constructing the model also helped us to refine the syntactic presentation of the type theory, since it pushed us towards replacing the \enquote{release} term by three more primitive constructs \enquote{pay}, \enquote{plet} and \enquote{split}.
We are confident that this is not the end of the story, and that our insights will help both in constructing more expressive type theories for reasoning about amortized cost, and in combining the categorical model presented here with categorical models for other effects and type system features.

\begin{acks}
This research was conducted as part of the project \enquote{Graded Modal Types for Quantitative Analysis of Higher-Order Probabilistic Programs} (\grantnum{ARIA}{MSAI-PR01-P11}) that was funded by the \grantsponsor{ARIA}{ARIA}{https://aria.org.uk/} programme on \enquote{Safeguarded AI}.
Dominic Orchard also received support from Schmidt Sciences, LLC.
\end{acks}

\section*{Data-Availability Statement}
\label{sec:data-availability}

This paper is accompanied by a formalization in the Lean proof assistant.
Theorems in the paper that have been formalized are marked with \lean.
The formalization will be published under an open source licence on Zenodo and will be submitted to artifact evaluation.

\bibliography{bibliography/bibliography.bib}

\newpage

\appendix

\section{Category Theory}
\label{sec:category-theory}
This appendix introduces the notation that we use for various categorical constructions in the paper.
The most important concept is that of a category.

\begin{definition}[Category]
    \label{def:appendix-categories:category}
    A category $\mathcal{C}$ consists of:
    \begin{enumerate}
        \item A collection of objects $\catobjects{\mathcal{C}}$.
        \item For every two objects $A,B \in \catobjects{\mathcal{C}}$ a collection of morphisms $\homset{\mathcal{C}}{A}{B}$.
        \item For every object $A \in \catobjects{\mathcal{C}}$ an identity morphism $\idmorph{A} \in \homset{\mathcal{C}}{A}{A}$.
        \item For objects $A,B,C \in \catobjects{\mathcal{C}}$ and morphisms $f \in \homset{\mathcal{C}}{A}{B}$ and $g \in \homset{\mathcal{C}}{B}{C}$
              the composition $g \circ f \in \homset{\mathcal{C}}{A}{C}$.
    \end{enumerate}
    The composition must be associative, and the identity morphism its unit.
\end{definition}
We need the standard categorical constructions of products and sums, as well as initial and terminal objects, in order to model the corresponding types in the type theory.

\begin{definition}[Categorical Product]
    \label{def:appendix-categories:product}
    Given two objects $A, B \in \catobjects{\mathcal{C}}$, then their categorical product $A \times B \in \catobjects{\mathcal{C}}$ comes with two morphims $\pi_1 : A \times B \to A$ and $\pi_2 : A \times B \to B$.
    For any other object $C$ in the category, if there are morphims $f: C \to A$ and $g : C \to B$, then there is a unique morphism $f \triangle g : C \to A \otimes B$ that makes the following diagram commute.

    \begin{center}
        \begin{tikzcd}
            & C \ar[dl, "f", swap] \ar[dr, "g"] \ar[d, dotted, "f \triangle g"]& \\
            A & A \times B \ar[l, "\pi_1"] \ar[r, "\pi_2", swap] & B \\
        \end{tikzcd}
    \end{center}
    We also need the parallel composition of arrows $f \times g$ which corresponds to the following commuting diagram:
    \begin{center}
        \begin{tikzcd}
            A \ar[d, "f"] & A \times B \ar[l,"\pi_1", swap] \ar[r, "\pi_2"] \ar[d, "f \times g", dotted] & B \ar[d,"g"] \\
            A' & A' \times B' \ar[l,"\pi_1", swap] \ar[r, "\pi_2"] & B' 
        \end{tikzcd}
    \end{center}
    which is defined as $f \times g = (f \circ \pi_1) \triangle (g \circ \pi_2)$.
\end{definition}

\begin{definition}[Categorical Sum]
    \label{def:appendix-categories:sum}
    Given two objects $A, B \in \catobjects{\mathcal{C}}$, then their categorical sum $A + B \in \catobjects{\mathcal{C}}$ comes with two morphims $\iota_1 : A \to A + B$ and $\iota_2 : B \to A + B$.
    For any other object $C$ in the category, if there are morphims $f: A \to C$ and $g : B \to C$, then there is a unique morphism $f \triangle g : C \to A \otimes B$ that makes the following diagram commute.

    \begin{center}
        \begin{tikzcd}
            & C  & \\
            A \ar[ur, "f"] \ar[r, "\iota_1", swap] & A + B \ar[u, dotted, "f \triangledown g", swap]& B \ar[ul, "g", swap] \ar[l, "\iota_2"]\\
        \end{tikzcd}
    \end{center}
    We also need the parallel composition of arrows $f + g$ which corresponds to the following commuting diagram:
    \begin{center}
        \begin{tikzcd}
          A \ar[d,"f"] \ar[r, "\iota_1"] & A + B \ar[d,"f + g", dotted] & B \ar[d,"g"] \ar[l, "\iota_2", swap] \\
          A' \ar[r, "\iota_1"] & A' + B' & B' \ar[l, "\iota_2", swap]
        \end{tikzcd}
    \end{center}
    which is defined as $f + g = (\iota_1 \circ f) \triangledown (\iota_2 \circ g)$.
\end{definition}

\begin{definition}[Initial and Terminal Element]
    Initial objects $0 \in \catobjects{\mathcal{C}}$ and terminal objects $1 \in \catobjects{\mathcal{C}}$ are characterized by existence of unique morphisms to, respectively from, any other object $C \in \catobjects{\mathcal{C}}$.
    \begin{center}
        \begin{tikzcd}
          0 \ar[r,"!", dotted] & C \ar[r, "!", dotted] & 1\\
        \end{tikzcd}
    \end{center}
\end{definition}

\begin{definition}[Monoidal Category]
    A monoidal category is a 6-tuple $(M, \otimes : M \times M \to M, I \in \catobjects{M}, \alpha, \lambda, \rho)$, $\otimes$ is a bifunctor, $I$ is its unit, and $\alpha$, $\lambda$ and $\rho$ are natural isomorphisms
    which witness the associativity, left and right unit laws.
    \[
    \begin{array}{rcrclcl}
        \alpha_{A,B,C} &:& (A \otimes B) \otimes C &\simeq& A \otimes (B \otimes C) &:& \alpha^{-1}_{A,B,C} \\
        \lambda_A &:& I \otimes A &\simeq& A &:& \lambda^{-1}_A \\
        \rho_A &:& A \otimes I &\simeq& A &:& \rho^{-1}_A
    \end{array}
    \]
    These natural transformations must fulfil coherence conditions, i.e. the following diagrams must commute.

    \begin{center}
      \begin{tikzcd}
        A \otimes (B \otimes (C \otimes D)) \ar[r, "\alpha_{A,B,C \otimes D}"] \ar[d, "\idmorph{A}\otimes (\alpha_{B,C,D})"] & (A \otimes B) \otimes (C \otimes D) \ar[r, "\alpha_{A \otimes B,C,D}"] & ((A \otimes B) \otimes C) \otimes D\\
        A \otimes ((B \otimes C) \otimes D) \ar[rr, "\alpha_{A, B\otimes C,D}"] & & (A \otimes (B \otimes C)) \otimes D \ar[u, "\alpha_{A,B,C} \otimes \idmorph{D}"]
      \end{tikzcd}
      \begin{tikzcd}
        A \otimes (I \otimes B) \ar[d, "\idmorph{A} \otimes \lambda_{B}"] \ar[r, "\alpha_{A,I,B}"] & (A \otimes I) \otimes B \ar[dl, "\rho_A \otimes \idmorph{B}"] \\
        A \otimes B
      \end{tikzcd}
    \end{center}
\end{definition}
\begin{definition}[Symmetric Monoidal Category]
    A monoidal category is called symmetric monoidal if there is a natural isomorphism $\sigma$ which witnesses the symmetry of the tensor $\otimes$.
    \[
    \begin{array}{rcrclcl}
        \sigma_{A,B} &:& A \otimes B &\simeq& B \otimes A &:& \sigma_{B,A}
    \end{array}
    \]
    For symmetric monoidal categories the following diagrams must commute:
    \begin{center}
        
        \begin{tikzcd}
          (A \otimes B) \otimes C \ar[r, "\sigma_{A,B} \otimes \idmorph{C}"] \ar[d, "\alpha_{A,B,C}"] & (B \otimes A) \otimes C \ar[d, "\alpha_{B,A,C}"] \\
          A \otimes (B \otimes C) \ar[d, "\sigma_{A, B\otimes C}"] & B \otimes (A \otimes C) \ar[d, "\idmorph{B} \otimes \sigma_{A,C}"] \\
          (B \otimes C) \otimes A \ar[r, "\alpha_{B,C,A}"] & B \otimes (C \otimes A) \\
        \end{tikzcd}

        \begin{tikzcd}
            A \otimes I \ar[r, "\sigma_{A,I}"] \ar[d, "\rho_{A}"] & I \otimes A \ar[ld, "\lambda_{A}"] \\
            A
        \end{tikzcd}
        \begin{tikzcd}
            A \otimes B \ar[d, "\idmorph{A \otimes B}"] \ar[r,"\sigma_{A,B}"] & B \otimes A \ar[dl, "\sigma_{B,A}"]\\
            A \otimes B
        \end{tikzcd}
    \end{center}
\end{definition}

\begin{remark}
    Categorical product and sum (cf. \cref{def:appendix-categories:product,def:appendix-categories:sum}) form symmetric monoidal categories.
\end{remark}

\begin{definition}[Strong Monoidal Functor]
    A strong monoidal functor $F : (\mathcal{C}, \otimes_\mathcal{C}, I_\mathcal{C}) \to (\mathcal{D}, \otimes_\mathcal{D}, I_\mathcal{D})$ is a functor $F : \mathcal{C} \to \mathcal{D}$
    together with:
    \[
    \begin{array}{rcrclcl}
        \epsilon &:& F(I_\mathcal{C}) &\simeq& I_\mathcal{D}  &:& \epsilon^{-1} \\
        \mu_{A,B} &:&  F(A) \otimes_\mathcal{D} F(B) &\simeq& F(A \otimes_\mathcal{C} B) &:& \mu^{-1}_{A,B} \\
    \end{array}
    \]
\end{definition}

\begin{lemma}[Existence of Ends]
    If $X$ is small-complete (all small diagrams in $X$ have limits in $X$) and $C$ is small, then every functor $S : C^\mathrm{op}\times C \to X$ has an end in $X$.
\end{lemma}
\begin{proof}
    See \citet[Corollary 2, p.~224]{Maclane1971}.
\end{proof}

\begin{definition}[Yoneda Embedding]
    The yoneda embedding $\yoneda{-} : \mathbb{C} \to \mathbb{X}$ is defined as:
    \begin{equation*}
        \yoneda{c} \coloneq \lambda p. \mathbf{Hom}_\mathbb{C}(c,p)
    \end{equation*}
\end{definition}

\begin{definition}[Day Convolution]
    Let $(\mathcal{C}, \otimes, I)$ be a monoidal category, and $[\mathcal{C},\setcategory]$ a category of Set-valued functors.
    Then the monoidal structure on $\mathcal{C}$ induces a monoidal structure on $[\mathcal{C},\setcategory]$, the day convolution.

    Let $F, G \in \catobjects{[\mathcal{C}, \setcategory]}$, then day convolution $\dayconvolution : [\mathcal{C}, \setcategory] \times [\mathcal{C}, \setcategory] \to [\mathcal{C}, \setcategory]$ is defined as follows:
    \begin{align*}
        F \dayconvolution G &\coloneq \int^{c_1, c_2} \yoneda{c_1 \otimes c_2} \times F(c_1) \times G(c_2) \\
        \dayidentity &\coloneq \yoneda{I}
    \end{align*}
    Then $([\mathcal{C}, \setcategory], \dayconvolution, \dayidentity)$ is a monoidal category.
\end{definition}

\begin{lemma}
    The Yoneda embedding $\yoneda{-} : (\mathcal{C}, \otimes, I) \to ([\mathcal{C}, \setcategory], \dayconvolution, \yoneda{I})$ is a strong monoidal functor.
\end{lemma}

\begin{definition}[Day Exponential]
    Similar to the familiar adjunction $- \times X \dashv -^{X}$ between categorical products and exponentials, there is an adjunction between day convolution and day exponentials. Day exponentials are right adjoint to Day convolution:
    \begin{equation*}
        - \dayconvolution X \dashv X \dayexponential -
    \end{equation*}
    This construction can be equivalently defined as follows:
    \begin{equation*}
        F \dayexponential G \coloneq  \lambda p. \int_{p_i} \mathbf{Hom}_{\mathrm{Set}}(F(p_i), G(p_i + p))
    \end{equation*}
\end{definition}

\begin{lemma}[Properties of Day Convolution and Day Exponential]
    \[
    \begin{array}{rcrclcl}
        \mathrm{uncurry} &:& X \dayexponential (Y \dayexponential Z) &\simeq& (X \dayconvolution Y) \dayexponential Z  &:& \mathrm{curry} \\
    \end{array}
    \]
    \label{lem:model:day-properties}
\end{lemma}

\begin{definition}[Exponential graded comonad]
  For a symmetric monoidal category $(\mathbb{C}, \otimes, I)$ and pre-ordered semiring $(\mathcal{R}, +, 0, \ast, 1, \leq)$ represented as a category
  semiringoidal category $\mathcal{R}$, an exponential graded comonad comprises a functor $\Box : \mathcal{R}^\mathsf{op} \to [\mathbb{C}, \mathbb{C}]$ along with natural transformations:
  \[
    \begin{array}{rcl}
      \varepsilon_A &:& \Box_1 A \to A \\
      \delta_{r,s,A} &:& \Box_{r \ast s} A \to \Box_r \Box_s A \\
      m_{r,A,B} &:& \Box_r A \times \Box_r B \to \Box_r (A \otimes B) \\
      m_{r,I} &:& I \rightarrow \Box_r I \\
      w_{A} &:& \Box_0 A \rightarrow I \\
      c_{r,s,A} &:& \Box_{r + s} A \rightarrow \Box_r A \otimes \Box_s A
    \end{array}
  \]
  satisfying a number of coherence conditions that essentially amount to ensuring the semiring-like properties of the
  operations above (see~\citet{Gaboardi2016combining}).
  Note that when $r \leq s$ we have a morphism $f : r \to s \in \mathcal{R}$ which lifts to a natural transformation
  $(\Box f)_A : \Box_s A \rightarrow \Box_r A$
\end{definition}

\section{Complete Description of the Equational Theory}
\label{sec:lambda-amor-core-complete}
This appendix contains the complete equational theory introduced in \cref{sec:lambda-amor-core}.

\subsection{Equivalence Relation}

Since it is an equivalence relation, we have to include the following three rules:

\begin{minipage}{0.25\textwidth}
    \eqrulerefl
\end{minipage}
\begin{minipage}{0.25\textwidth}
    \eqrulesymm
\end{minipage}
\begin{minipage}{0.4\textwidth}
    \eqruletrans
\end{minipage}

\subsection{Congruence Rules}

The equational theory should be a congruence relation, which means that there should be one congruence rule for each term-forming operation.

\begin{minipage}{0.45\textwidth}
    \eqruleconginl
\end{minipage}
\begin{minipage}{0.45\textwidth}
     \eqruleconginr
\end{minipage}

\begin{minipage}{0.45\textwidth}
    \eqrulecongfst
\end{minipage}
\begin{minipage}{0.45\textwidth}
    \eqrulecongsnd
\end{minipage}

\begin{prooftree}
    \AxiomC{$\mathcal{D}_1 \equiv \mathcal{D}'_1$}
    \AxiomC{$\mathcal{D}_2 \equiv \mathcal{D}'_2$}
    \RightLabel{$\equiv$\textsc{-Cong-Case-UnitP}}
    \BinaryInfC{$\textsc{Case-UnitP}(\mathcal{D}_1, \mathcal{D}_2) \equiv \textsc{Case-UnitP}(\mathcal{D}'_1, \mathcal{D'}_2)$}
\end{prooftree}

\begin{minipage}{0.45\textwidth}
    \begin{prooftree}
        \AxiomC{$\mathcal{D}_1 \equiv \mathcal{D}_2$}
        \RightLabel{$\equiv$\textsc{-Cong-Void}}
        \UnaryInfC{$\textsc{Void}(\mathcal{D}_1) \equiv \textsc{Void}(\mathcal{D}_2)$}
    \end{prooftree}
\end{minipage}
\begin{minipage}{0.45\textwidth}
    \begin{prooftree}
        \AxiomC{$\mathcal{D}_1 \equiv \mathcal{D}_2$}
        \RightLabel{$\equiv$\textsc{-Cong-Abs}}
        \UnaryInfC{$\textsc{Abs}(\mathcal{D}_1) \equiv \textsc{Abs}(\mathcal{D}_2)$}
    \end{prooftree}
\end{minipage}

\begin{prooftree}
    \AxiomC{$\mathcal{D}_1 \equiv \mathcal{D}'_1$}
    \AxiomC{$\mathcal{D}_2 \equiv \mathcal{D}'_2$}
    \RightLabel{$\equiv$\textsc{-Cong-App}}
    \BinaryInfC{$\textsc{Case-App}(\mathcal{D}_1, \mathcal{D}_2) \equiv \textsc{Case-App}(\mathcal{D}'_1, \mathcal{D}'_2)$}
\end{prooftree}

\begin{prooftree}
    \AxiomC{$\mathcal{D}_1 \equiv \mathcal{D}'_1$}
    \AxiomC{$\mathcal{D}_2 \equiv \mathcal{D}'_2$}
    \RightLabel{$\equiv$\textsc{-Cong-Tensor}}
    \BinaryInfC{$\textsc{Tensor}(\mathcal{D}_1, \mathcal{D}_2) \equiv \textsc{Tensor}(\mathcal{D}'_1, \mathcal{D}'_2)$}
\end{prooftree}

\begin{prooftree}
    \AxiomC{$\mathcal{D}_1 \equiv \mathcal{D}'_1$}
    \AxiomC{$\mathcal{D}_2 \equiv \mathcal{D}'_2$}
    \RightLabel{$\equiv$\textsc{-Cong-Case-Tensor}}
    \BinaryInfC{$\textsc{Case-Tensor}(\mathcal{D}_1, \mathcal{D}_2) \equiv \textsc{Case-Tensor}(\mathcal{D}'_1, \mathcal{D}'_2)$}
\end{prooftree}

\begin{prooftree}
    \AxiomC{$\mathcal{D}_1 \equiv \mathcal{D}'_1$}
    \AxiomC{$\mathcal{D}_2 \equiv \mathcal{D}'_2$}
    \AxiomC{$\mathcal{D}_3 \equiv \mathcal{D}'_3$}
    \RightLabel{$\equiv$\textsc{-Cong-Case-Sum}}
    \TrinaryInfC{$\textsc{Case-Sum}(\mathcal{D}_1, \mathcal{D}_2, \mathcal{D}_3) \equiv \textsc{Case-Sum}(\mathcal{D}'_1, \mathcal{D}'_2, \mathcal{D}'_3)$}
\end{prooftree}

\begin{prooftree}
    \AxiomC{$\mathcal{D}_1 \equiv \mathcal{D}'_1$}
    \AxiomC{$\mathcal{D}_2 \equiv \mathcal{D}'_2$}
    \RightLabel{$\equiv$\textsc{-Cong-With}}
    \BinaryInfC{$\textsc{With}(\mathcal{D}_1, \mathcal{D}_2) \equiv \textsc{With}(\mathcal{D}'_1, \mathcal{D}'_2)$}
\end{prooftree}

\begin{minipage}{0.45\textwidth}
    \begin{prooftree}
        \AxiomC{$\mathcal{D}_1 \equiv \mathcal{D}_2$}
        \RightLabel{$\equiv$\textsc{-Cong-Ret}}
        \UnaryInfC{$\textsc{Ret}(\mathcal{D}_1) \equiv \textsc{Ret}(\mathcal{D}_2)$}
    \end{prooftree}
\end{minipage}
\begin{minipage}{0.45\textwidth}
    \begin{prooftree}
        \AxiomC{$\mathcal{D}_1 \equiv \mathcal{D}_2$}
        \RightLabel{$\equiv$\textsc{-Cong-Store}}
        \UnaryInfC{$\textsc{Store}(\mathcal{D}_1) \equiv \textsc{Store}(\mathcal{D}_2)$}
    \end{prooftree}
\end{minipage}

\begin{prooftree}
    \AxiomC{$\mathcal{D}_1 \equiv \mathcal{D}'_1$}
    \AxiomC{$\mathcal{D}_2 \equiv \mathcal{D}'_2$}
    \RightLabel{$\equiv$\textsc{-Cong-Bind}}
    \BinaryInfC{$\textsc{Bind}(\mathcal{D}_1, \mathcal{D}_2) \equiv \textsc{Bind}(\mathcal{D}'_1, \mathcal{D}'_2)$}
\end{prooftree}

\begin{prooftree}
    \AxiomC{$\mathcal{D}_1 \equiv \mathcal{D}'_1$}
    \AxiomC{$\mathcal{D}_2 \equiv \mathcal{D}'_2$}
    \RightLabel{$\equiv$\textsc{-Cong-Release}}
    \BinaryInfC{$\textsc{Release}(\mathcal{D}_1, \mathcal{D}_2) \equiv \textsc{Release}(\mathcal{D}'_1, \mathcal{D}'_2)$}
\end{prooftree}

\subsection{Applications of Subsumption}

We have to permute the subsumption rule through type formers:
\begin{align*}
    \tag*{$\equiv$\textsc{-Sub-Inl}}
    \textsc{Inl}(\textsc{Sub}(\mathcal{D}, \mathcal{X}, \mathcal{S})) &\equiv \textsc{Sub}(\textsc{Inl}(\mathcal{D}), \mathcal{X}, \oplus(\mathcal{S}, \textsc{Refl})) \\
    \tag*{$\equiv$\textsc{-Sub-Inr}}
    \textsc{Inr}(\textsc{Sub}(\mathcal{D}, \mathcal{X}, \mathcal{S})) &\equiv \textsc{Sub}(\textsc{Inr}(\mathcal{D}), \mathcal{X}, \oplus(\textsc{Refl}, \mathcal{S})) \\
    \tag*{$\equiv$\textsc{-Sub-Tensor}}
    \textsc{Tensor}(\textsc{Sub}(\mathcal{D}_1,\mathcal{X}_1,\mathcal{S}_1), \textsc{Sub}(\mathcal{D}_2, \mathcal{X}_2, \mathcal{S}_2)) &\equiv \textsc{Sub}(\textsc{Tensor}(\mathcal{D}_1, \mathcal{D}_2), (\mathcal{X}_1, \mathcal{X}_2), \otimes(\mathcal{S}_1, \mathcal{S}_2)) \\
    \tag*{$\equiv$\textsc{-Sub-With}}
    \textsc{With}(\textsc{Sub}(\mathcal{D}_1, \mathcal{X}, \mathcal{S}_1), \textsc{Sub}(\mathcal{D}_2, \mathcal{X}, \mathcal{S}_2)) &\equiv \textsc{Sub}(\textsc{With}(\mathcal{D}_1,\mathcal{D}_2), \mathcal{X}, \with(\mathcal{S}_1,\mathcal{S}_2)) \\
    \tag*{$\equiv$\textsc{-Sub-Fst}}
    \textsc{Fst}(\textsc{Sub}(\mathcal{D}, \mathcal{X}, \with(\mathcal{S}_1, \mathcal{S}_2))) &\equiv \textsc{Sub}(\textsc{Fst}(\mathcal{D}),\mathcal{X}, \mathcal{S}_1) \\
    \tag*{$\equiv$\textsc{-Sub-Snd}}
    \textsc{Snd}(\textsc{Sub}(\mathcal{D}, \mathcal{X}, \with(\mathcal{S}_1, \mathcal{S}_2))) &\equiv \textsc{Sub}(\textsc{Snd}(\mathcal{D}), \mathcal{X}, \mathcal{S}_2) \\
    \tag*{$\equiv$\textsc{-Case-UnitP}}
    \textsc{Case-UnitP}(\textsc{Sub}(\mathcal{D}_1, \mathcal{X}_1,\textsc{Refl}), \textsc{Sub}(\mathcal{D}_2), \mathcal{X}_2, \mathcal{S})  &\equiv \textsc{Sub}(\textsc{Case-UnitP}(\mathcal{D}_1, \mathcal{D}_2),(\mathcal{X}_1, \mathcal{X}_2), \mathcal{S}) \\
    \tag*{$\equiv$\textsc{-Sub-Ret}}
    \textsc{Ret}(\textsc{Sub}(\mathcal{D}, \mathcal{X}, \mathcal{S})) &\equiv \textsc{Sub}(\textsc{Ret}(\mathcal{D}), \mathcal{X}, \costmonad{\leq}{\mathcal{S}})
\end{align*}
and we can eliminate uses of the subsumption rule if we use it with a subtyping derivation for reflexivity.
Similarly, applications of \textsc{Trans} can be replaced by two uses of the subsumption rule instead.
\begin{align*}
    \tag*{$\equiv$\textsc{-Sub-Refl}}
    \textsc{Sub}(\mathcal{D},\textsc{Refl}, \textsc{Refl}) &\equiv \mathcal{D} \\
    \tag*{$\equiv$\textsc{-Sub-Trans}}
    \textsc{Sub}(\mathcal{D}, \textsc{Trans}(\mathcal{X}_1, \mathcal{X}_2), \textsc{Trans}(\mathcal{S}_1,\mathcal{S}_2)) &\equiv \textsc{Sub}(\textsc{Sub}(\mathcal{D},\mathcal{X}_1, \mathcal{S}_1), \mathcal{X}_2, \mathcal{S}_2) \\
\end{align*}

\subsection{Computation Rules}

Then we have to add the following $\beta$-rules, where we use the derived substitution operation from \cref{def:lamor:substitution-derivations} in some of the rules.
\begin{align}
    \tag*{$\equiv$-$\beta$-$\textsc{Fun}$}
    \textsc{App}(\textsc{Abs}(\mathcal{D}_1), \mathcal{D}_2) &\equiv \textsc{Subst}_1(\mathcal{D}_1, \mathcal{D}_2) \\
    \tag*{$\equiv$-$\beta$-$\textsc{Unit-P}$}
    \textsc{Case-Unit-P}(\textsc{Unit-P}, \mathcal{D}) &\equiv \mathcal{D} \\
    \tag*{$\equiv$-$\beta$-$\textsc{With}_1$}
    \textsc{Fst}(\textsc{With}(\mathcal{D}_1, \mathcal{D}_2)) &\equiv \mathcal{D}_1 \\
    \tag*{$\equiv$-$\beta$-$\textsc{With}_2$}
    \textsc{Snd}(\textsc{With}(\mathcal{D}_1, \mathcal{D}_2)) &\equiv \mathcal{D}_2 \\
    \tag*{$\equiv$-$\beta$-$\textsc{Sum}_1$}
    \textsc{Case-Sum}(\textsc{Inl}(\mathcal{D}_1), \mathcal{D}_2, \mathcal{D}_3) &\equiv \textsc{Subst}_1(\mathcal{D}_2, \mathcal{D}_1) \\
    \tag*{$\equiv$-$\beta$-$\textsc{Sum}_2$}
    \textsc{Case-Sum}(\textsc{Inr}(\mathcal{D}_1), \mathcal{D}_2, \mathcal{D}_3) &\equiv \textsc{Subst}_1(\mathcal{D}_3, \mathcal{D}_1) \\
    \tag*{$\equiv$-$\beta$-$\textsc{Tensor}$}
    \textsc{Case-Tensor}(\textsc{Tensor}(\mathcal{D}_1, \mathcal{D}_2), \mathcal{D}_3) &\equiv \textsc{Subst}_2(\mathcal{D}_3, \mathcal{D}_1, \mathcal{D}_2)
\end{align}

\subsection{Extensionality Rules}
We include $\eta$-equalities for all negative types.

\begin{align}
    \tag*{$\equiv$-$\eta$\textsc{-With}}
    \textsc{With}(\textsc{Fst}(\mathcal{D}), \textsc{Snd}(\mathcal{D})) &\equiv \mathcal{D} \\
    \tag*{$\equiv$-$\eta$\textsc{-Fun}}
    \textsc{Abs}(\textsc{App}(\mathcal{D}, \textsc{Var})) &\equiv \mathcal{D} \\
    \tag*{$\equiv$-$\eta$\textsc{-Unit-N}}
    \textsc{Unit-N} &\equiv \mathcal{D}
\end{align}

Note that there is a side condition for the rule $\equiv$-$\eta$\textsc{-Fun} regarding free variables which cannot be properly expressed using this shorthand notation.

\section{Detail and Proofs for the Categorical Model}
\label{sec:model-details}
\subsection{Commuting Diagrams for Strong Graded Monads}
\label{apx:model-details:strong-graded-monad-diagrams}

The following diagrams, omitted from \cref{def:adjoint-model:strong-graded-monad}, give the full coherence conditions
for a strong graded monad, i.e, that the strong natural transformation $\mathit{st}_{p,X,Y} : X \otimes M(p,Y) \to M(p, X \otimes Y)$ must satisfy with
respect to the unit $\eta$ and multiplication $\mu$ of the graded monad:

\begin{center}
  \begin{tikzcd}
    X \otimes Y \ar[r, "\eta_Y"] \ar[d, "{\idmorph{} \otimes \eta_Y}" swap] & M(0,X \otimes Y) \\
    X \otimes M(0,Y) \ar[ru, "{\mathit{st}_{0,X,Y}}" swap] & \\
  \end{tikzcd}
  \begin{tikzcd}
    M(p,X) \ar[r, "M{(p,\lambda_X)}"] \ar[d, "\lambda_{M(p, X)}" swap] & M(p,1 \otimes X) \\
    1 \otimes M(p,X) \ar[ur, "{\mathit{st}_{p,1,X}}"] & \\
  \end{tikzcd}
\end{center}

\begin{center}
  \begin{tikzcd}
    (X \otimes Y) \otimes M(p,Z) \ar[rr, "{\mathit{st}_{p, X \otimes Y, Z}}"] \ar[d, "\alpha" swap] & & M(p,(X \otimes Y) \otimes Z) \ar[d, "M\alpha"] \\
    X \otimes (Y \otimes M(p,Z)) \ar[r, "\idmorph{} \otimes {\mathit{st}_{p,Y,Z}}" swap] & X \otimes M(p,Y \otimes Z) \ar[r, "{\mathit{st}_{p,X,Y \otimes Z}}" swap] & M(p,X \otimes (Y \otimes Z)) \\
  \end{tikzcd}
\end{center}
\begin{center}
  \begin{tikzcd}
    X \otimes M(p_1,M(p_2,Y)) \ar[r, "{\mathit{st}_{p_1,X,M(p_2,Y)}}"] \ar[d, "{\idmorph{} \otimes \mu_{p_1,p_2,Y}}" swap] & M(p_1, X \otimes M(p_2,Y)) \ar[r, "{M(p_1, \mathit{st}_{p_2,X,Y})}"] & M(p_1,M(p_2, X \otimes Y)) \ar[d, "{\mu_{p_1,p_2,X \otimes Y}}"] \\
    X \otimes M(p_1 + p_2, Y) \ar[rr, "{\mathit{st}_{p_1 + p_2, X,Y}}" swap] & & M(p_1 + p_2, X \otimes Y) \\
  \end{tikzcd}
\end{center}

\subsection{Proof of \cref{lem:adjoint-model:subtyping}}
\label{apx:model-details:subtyping}

\subtypinginmodel*
\begin{proof}
    We show the first part of this lemma by induction on the subtyping derivation $\mathcal{S}$; we show a subset of the relevant cases.
    The extension to subtyping derivations between typing contexts is straightforward.
    \begin{description}
        \item[S-Refl] We need to construct a morphism from $\interpret{\tau}$ to $\interpret{\tau}$; we can just pick $\idmorph{\interpret{\tau}}$.
        \item[S-Trans] We assume there are subtyping derivations $\mathcal{S}_1 : \tau_1 \sub \tau_2$ and $\mathcal{S}_2 : \tau_2 \sub \tau_3$.
        By the induction hypothesis there exist morphisms $\interpret{\mathcal{S}_1} : \interpret{\tau_1} \to \interpret{\tau_2}$ and $\interpret{\mathcal{S}_2} : \interpret{\tau_2} \to \interpret{\tau_3}$.
        We can thus use $\interpret{\mathcal{S}_2} \circ \interpret{\mathcal{S}_1} : \interpret{\tau_1} \to \interpret{\tau_3}$.
        \item[S-Tensor] We assume there are subtyping derivations $\mathcal{S}_1 : \tau_1 \sub \tau_2$ and $\mathcal{S}_2 : \tau_3 \sub \tau_4$.
        By the induction hypothesis there exist morphisms $\interpret{\mathcal{S}_1} : \interpret{\tau_1} \to \interpret{\tau_2}$ and $\interpret{\mathcal{S}_2} : \interpret{\tau_3} \to \interpret{\tau_4}$.
        We need to construct a morphism from $\interpret{\tau_1} \otimes \interpret{\tau_3}$ to $\interpret{\tau_2} \otimes \interpret{\tau_4}$.
        Since $\otimes : \modelcategory \times \modelcategory \to \modelcategory$ is a bifunctor we can simply choose $\interpret{\mathcal{S}_1} \otimes \interpret{\mathcal{S}_2}$.
        \item[S-Potential] We assume that there is a subtyping derivation $\mathcal{S} : \tau_1 \leq \tau_2$ and a morphism $(p_2 \leq p_1) : p_2 \to p_1$.
        By the induction hypothesis there exists a morphism $\interpret{\mathcal{S}} : \interpret{\tau_1}  \to \interpret{\tau_2}$.
        Since $P$ is a functor that is contravariant in its first argument and covariant in its second argument we can choose $P(p_2 \leq p_1, \interpret{\mathcal{S}})$.
        \item[S-Monad] We assume that there is a subtyping derivation $\mathcal{S} : \tau_1 \leq \tau_2$ and a morphism $(p_1 \leq p_2) : p_1 \to p_2$.
        By the induction hypothesis there exists a morphism $\interpret{\mathcal{S}} : \interpret{\tau_1}  \to \interpret{\tau_2}$.
        Since $M$ is a functor that is covariant in both of its arguments we can pick $M(p_1\leq p_2, \interpret{\mathcal{S}})$.
    \end{description}
\end{proof}

\end{document}